\documentclass[11pt]{article}
\usepackage[margin=1in]{geometry}
\usepackage{amsmath,amssymb,amsthm}
\usepackage{booktabs}
\usepackage{array}
\usepackage{float}
\usepackage{graphicx}
\usepackage{natbib}
\usepackage[hidelinks]{hyperref}
\hypersetup{pdftitle={Information Set Emulation: Causal Certificates for AI Derived EHR Features},
  pdfauthor={Takes Fujita and Nobutaka Hattori}}
\ifdefined\pdfgentounicode
  \input{glyphtounicode}
\fi
\usepackage{enumitem}
\setlist[itemize]{leftmargin=1.2em}
\setlist[enumerate]{leftmargin=1.4em}

\newtheorem{definition}{Definition}
\newtheorem{theorem}{Theorem}
\newtheorem{proposition}{Proposition}
\newtheorem{corollary}{Corollary}

\title{Information Set Emulation: Causal Certificates for AI Derived EHR Features}
\author{Takes Fujita\textsuperscript{1} \quad Nobutaka Hattori\textsuperscript{2}\\[4pt]
\small \textsuperscript{1}VRI\\
\small \textsuperscript{2}Department of Neurology, Juntendo University School of Medicine}
\date{}

\begin{document}
\maketitle

\begin{abstract}
AI and large language models can recover clinically meaningful features from electronic health records (EHRs), but predictive usefulness does not establish admissibility for causal inference. We introduce information set emulation: an AI typed lift attaches source evidence, clinical and recording times, decision-time availability, representation version, proposed causal roles, and unresolved ambiguity to extracted features under a locked target trial. Causal certificates record auditable evidence for those roles. Features with unresolved downstream roles are routed to compatible reporting or separate analyses.

Typed evidence defines an observational fiber of causal worlds consistent with the observed law. The locked scalar estimand maps this fiber to a compatible image whose squared Chebyshev radius equals the residual minimax mean squared error when the image is nonempty and compact. This classical identity provides a target-specific measure of information ambiguity. The contribution is its integration with a joint EHR observation map and an auditable certificate architecture.

Under explicit exchangeability, positivity, and nuisance-consistency conditions, we give identification and cross-fitted augmented inverse probability weighted estimation, distinguishing empirical and population targets. An EHR compression-drift identity separates the roles of frame presence, treatment assignment, and outcome observation. Artificial simulations and a common-law finite-world example illustrate estimation failures and information-radius reduction. Synthetic Phase 0 notes demonstrate audit diagnostics; a separate role-specific analysis spread illustrates routing and is not an exact fiber radius. All experiments are synthetic. The framework specifies when reconstructed information can support a point claim and when compatible reporting is required.

\end{abstract}


\section{Introduction}

Electronic health records have changed the scale of clinical causal research. They contain diagnoses, medications, laboratory results, vital signs, radiology reports, clinician notes, nursing documentation, orders, referrals, encounter patterns, and administrative records. They also contain traces of clinical reasoning that are rarely encoded as clean variables. Frailty may appear in a note. Performance status may be implied by mobility descriptions, goals of care, oxygen requirement, or the decision not to offer an aggressive treatment. Contraindications may appear in free text rather than in problem lists. Treatment choice may depend on a clinician judgement that is partly recorded, partly tacit, and partly institutional.

This creates a familiar problem for retrospective causal inference. The target trial may be well specified, yet the information set that would have been measured before treatment in a prospective study is only partially available in the EHR. Standard structured variables may omit the very reasons why one treatment was chosen over another. AI and large language models can extract clinically interpretable features from free text and irregular event streams. They can recover performance status, frailty, goals of care, symptom burden, treatment rationale, or care context that would otherwise be unmeasured. Used well, such features can reduce the distance between a retrospective EHR study and the prospective study it seeks to emulate.

Used poorly, the same features can make the causal analysis worse. A feature extracted from a discharge summary may look like baseline frailty while encoding the outcome. A post treatment progress note may contain treatment response. A laboratory trajectory may reflect the decision to monitor more closely. A note density score may measure observation intensity rather than disease state. A model selected because it improves outcome prediction may produce an effect modifier that cannot support a confirmatory claim. The issue is not whether AI is useful. It is whether the information it reconstructs can identify the causal effect being estimated.

This paper proposes information set emulation. Target trial emulation asks whether the observational study has emulated the trial protocol. Information set emulation asks whether the retrospective data, augmented by AI, have emulated the pre treatment information set that a prospective study would have used for eligibility, treatment choice, observation, and outcome comparison. The distinction matters. Two studies can emulate the same treatment strategies, time zero, outcome window, and target population, yet differ sharply in the information available for confounding control, follow up, and interpretation.

The method starts with AI, but it does not treat AI as an oracle. The first step is not outcome modeling, propensity score estimation, or generic feature extraction. The first step is causal sorting of the EHR event stream. Given the target trial lock, AI reads the raw event stream and returns a typed causal information ledger: candidate features, their source evidence, the time they refer to, when they were recorded, whether they were available at the treatment decision, their proposed causal role, and unresolved ambiguity. This lift is the step that can make the retrospective record resemble a prospective information set. Its value is measured by how much it shrinks the compatible estimand set, not by prediction accuracy alone. The certificate and information radius then decide whether the remaining ambiguity is small enough to support causal effect estimation.

The core claim is therefore stronger than a data cleaning recommendation. Without sufficient evidence about source time, availability, and causal role, the prospective causal effect need not be identified from the flat EHR table. The obstruction is nonidentification when compatible worlds imply different values of the locked estimand. With enough typed evidence, the compatible set of prospective estimands shrinks. If the remaining radius is small relative to the scientific tolerance and statistical error, a retrospective EHR study can approximate the causal precision of a prospective study for the locked claim.

Information set emulation uses a causal typing principle: an AI derived feature must have a role relative to the locked target trial before it can enter causal estimation. The EHR problem differs from ordinary generated covariate use because the relevant object is not a single generated variable. It is a reconstructed pre treatment information set whose source time, availability, and causal role are only partially observed. The paper therefore moves from feature extraction to information structure reconstruction. EHRs are a demanding case, but the object is broader: a retrospective event stream must be organized into typed evidence before it can emulate the information structure of a prospective design.

The paper has three contributions. The first is constructive. It defines the AI typed lift and the causal certificate as an auditable way to turn a retrospective event stream into typed evidence for a locked target trial. This includes the downstream domination rule, which keeps features with unresolved mediator, outcome proxy, observation process, or mixed-window roles out of the primary point estimator while retaining them for compatible reporting and sensitivity analysis.

The second contribution is a decision-theoretic formulation specialized to typed EHR evidence. Typed evidence induces an observational fiber of prospective causal worlds. The locked estimand maps that fiber to a compatible estimand image, and the squared Chebyshev radius of this image is the exact minimax mean squared error for the locked scalar target. AI typing and additional certification are therefore refinements of the information structure: they are valuable when they shrink this radius without excluding plausible causal worlds.

The third contribution connects the information structure to estimation and practice. When the radius is small enough and the usual causal identification assumptions are credible, the paper gives an operational EHR causal effect and an augmented inverse probability weighted estimator. It also applies the established conditional-covariance compression identity to the joint EHR observation mechanism, defines a certificate-conditional approximation frontier, and uses artificial EHR simulations, stress tests, and a synthetic Phase~0 note validation to show when AI typed reconstruction succeeds, fails, or should stop at compatible reporting.

The paper is a theory and simulation study. It does not use patient data and does not claim a clinical treatment effect. Several mathematical ingredients are classical; the contribution is the certificate architecture, the joint EHR observation map, the role-aware compatible estimand construction, and their operational connection to validation and stopping rules. The aim is to define when AI reconstructed EHR information can support causal effect identification and estimation, and when it cannot.
\section{EHR event streams and prospective information}

Let $i=1,\ldots,n$ index patients. A patient may contribute one or more eligible treatment episodes; for clarity we begin with one episode per patient and return to repeated episodes later. Let $T_i$ be the index time or time zero. Let $A_i\in\{0,1\}$ denote the treatment strategy initiated at $T_i$, and let $Y_i$ be an outcome measured over a fixed follow up window. Let $Y_i(a)$ be the potential outcome under treatment strategy $a$.

The EHR is an event stream
\[
E_i=\{e_{iq}:q=1,\ldots,Q_i\}.
\]
Each event is represented as
\[
e_{iq}=\left(t^{clin}_{iq},t^{meas}_{iq},t^{rec}_{iq},t^{avail}_{iq},b_{iq},x_{iq}\right),
\]
where $t^{clin}$ is the clinical time referred to by the event, $t^{meas}$ is the measurement time, $t^{rec}$ is the time at which the event is recorded in the EHR, $t^{avail}$ is the time at which the information was available for decision making, $b$ is the event type, and $x$ is the observed content. These clocks need not agree. A note written after treatment may describe symptoms before treatment. A laboratory value measured before treatment may become available only after treatment. A discharge summary may contain both baseline history and outcome adjacent interpretation.

Let $\mathcal F_P^-$ denote the prospective pre treatment information set that a well designed prospective study would have measured or fixed before treatment. It includes the target trial eligibility information, clinical state, history, treatment availability, design variables, and decision relevant judgement that would have been available at $T_i$. The prospective information set is not necessarily observed in the retrospective EHR.

Let $\Lambda_\theta$ denote the AI typed lift applied before causal estimation. Given the EHR stream and the locked target trial, it returns a typed causal information ledger
\[
\Lambda_\theta(E_i,L_T)=
\left\{Z_{ij},\mathcal S_{ij},\widehat{\mathcal T}_{ij},K_{ij},u_{ij}:j=1,\ldots,p_i\right\}.
\]
Here $Z_{ij}$ is an extracted feature; $\mathcal S_{ij}$ is the source evidence, such as note span, order, laboratory record, or event segment; $\widehat{\mathcal T}_{ij}$ is the proposed set of causal roles; $K_{ij}$ is the certificate containing timing, availability, source, version, and role evidence; and $u_{ij}$ records unresolved ambiguity. The lift may be produced by an LLM, a rule based NLP system, a conventional model, a human AI workflow, or a manual review protocol. It is called a lift because it turns raw EHR records into a typed information object rather than a flat feature matrix.

Let $E_i^-(K)$ be the part of the EHR event stream declared by the typed ledger to be admissible for reconstructing pre treatment information. The resulting reconstructed vector can be written
\[
\widetilde V_i=R_\theta\{E_i^-(K),\xi_i\},
\]
where $\xi_i$ records stochastic generation, prompting, sampling, or extraction randomness.

The important object is not a vector of AI values but a typed information set. Compatibility and certification are distinct. For patient $i$ and feature $j$, define
\[
\mathcal T_{ij}(K)=\{t\in\mathcal T:\operatorname{Comp}_t(K_{ij},L_T)=1\},
\qquad
\mathcal C_{ij}^{cert}(K)=\{t\in\mathcal T:\operatorname{Cert}_t(K_{ij},L_T)=1\}.
\]
Here $\operatorname{Comp}_t=1$ means that the declared evidence and model have not ruled out role $t$; $\operatorname{Cert}_t=1$ means that the role-specific evidence requirements are met. Certification implies compatibility, so $\mathcal C_{ij}^{cert}(K)\subseteq\mathcal T_{ij}(K)$. Failure to certify a role does not rule it out. The predicates and the handling of missing evidence must be declared before estimation.
Let
\[
\mathcal A_{pre}=\{B,D,H,U\},\qquad
\mathcal A_{down}=\{M,Y,\mathsf{Obs},I,\mathsf{Excluded}\}.
\]
The patient-specific routing sets are
\[
J_{point,i}(K)=\{j:\mathcal T_{ij}(K)\neq\emptyset,\ 
\mathcal T_{ij}(K)\subseteq\mathcal A_{pre},\
\mathcal C_{ij}^{cert}(K)\cap\mathcal A_{pre}\neq\emptyset\},
\]
and
\[
J_{amb,i}(K)=\{j:\mathcal T_{ij}(K)\cap\mathcal A_{pre}\neq\emptyset,\
\mathcal T_{ij}(K)\cap\mathcal A_{down}\neq\emptyset\}.
\]
Features with both pre-treatment and remaining downstream roles are excluded from the primary point estimator. For a declared common feature schema, let $H_i$ contain the admitted values and their admissibility and missing-value indicators, and define
\[
\mathcal H_{\theta,K,i}^{point}=\sigma(H_i).
\]
All components of $H_i$, including routing and missing-value indicators, must satisfy pre-treatment measurability for a point analysis. An outcome-informed routing decision is not made admissible by masking the value alone. When the patient index is suppressed below, $\mathcal H_{\theta,K}^{point}$, $J_{point}(K)$, $J_{amb}(K)$, and $\mathcal T_j(K)$ denote these patient-specific objects under the common locked rule.

Ambiguous features remain available for compatible reporting of the same locked estimand across remaining causal worlds. Analyses that instead change the estimand, such as a direct-effect analysis, require separate locks and separate reports; their estimates are not automatically members of the original compatible estimand image. This downstream domination rule prevents a possible outcome proxy or mediator from entering the primary point estimator merely because a baseline interpretation is also plausible.

A vector of values is therefore not an information set. It becomes one only after the AI typed lift has attached source windows, clocks, representation version, causal role claims, and unresolved ambiguity, and after unresolved downstream roles have been separated from the point estimation set.

The lift is allowed to be stochastic before it is locked. We write
\[
K_\theta=\Lambda_\theta(E,L_T,\xi)
\]
for the evidence produced by the AI system, where $\xi$ denotes seed, sampling, prompt, fold, or model randomness declared before outcome analysis. After the lift is produced and locked, inference is conditional on the realized evidence object, or else averages over the declared distribution of $\xi$. The stochasticity is not the central difficulty. The central difficulty is that the same flat table can remain compatible with several prospective information structures. The typed lift is useful only insofar as it removes incompatible information structures without silently removing plausible ones.

The construction is short. The AI produces typed evidence $K_\theta$. That evidence refines the observational fiber of causal worlds compatible with the EHR record. The locked estimand maps this fiber to a compatible estimand set $\mathfrak P(K_\theta)$, whose squared Chebyshev radius $r^2(K_\theta)$ is the exact minimax mean squared error of that information structure for the scalar target. AI typing is valuable when it reduces that radius; additional certification is valuable when it reduces it further.

\begin{definition}[Information set emulation]
A retrospective event stream study, including an EHR study, performs information set emulation for a target trial if it specifies a prospective pre treatment information set $\mathcal F_P^-$ and constructs an auditable retrospective information set $\mathcal H_K$ intended to approximate the components of $\mathcal F_P^-$ that are needed for treatment assignment, eligibility, target construction, outcome observation, and potential outcome prediction.
\end{definition}

This paper develops the framework through EHR target trial emulation because EHRs are a demanding case: clinical facts, recording times, availability, observation processes, and outcome adjacent narratives are routinely entangled. The definition deliberately includes treatment assignment, eligibility, and observation processes. Outcome prediction alone is insufficient. A reconstructed feature may predict $Y$ well but fail to retain the information that determined treatment or follow up.

\section{AI typed EHR lift and causal certificates}

The typed lift is the operational engine of information set emulation. It performs the first causal sorting of the raw EHR: likely baseline state, design or decision information, prior history, candidate effect modifier, observation process, mediator, outcome proxy, intercurrent event, and excluded or unresolved material. This is not a claim that AI is an oracle. It is a claim that large scale retrospective causal inference requires a first pass that is too granular for ordinary manual variable selection. AI can supply that pass when it returns evidence along with types.

An AI derived EHR feature is therefore a claim about clinical information. A label such as frailty score, performance status, physician concern, or social vulnerability is not enough. The feature requires a certificate. A certificate does not certify truth. It records auditable evidence that a feature is admissible for a proposed causal role under a locked target trial. If the AI typed lift is strong, many roles become certified and the retrospective information set approaches the prospective information set. If the lift is weak, the compatible estimand set remains wide.

\begin{definition}[Causal certificate]
For a generated feature $Z_{ij}$, a causal certificate is
\[
K_{ij}=\left(\widehat T_{ij},\mathcal W^{clin}_{ij},\mathcal W^{meas}_{ij},\mathcal W^{rec}_{ij},\mathcal W^{avail}_{ij},R_{\theta j},\mathcal E_{ij},\mathcal G_{ij},\mathcal M_{ij}\right).
\]
Here $\widehat T_{ij}$ is the proposed causal type; the four $\mathcal W$ terms are the clinical, measurement, recording, and availability windows; $R_{\theta j}$ is the representation version; $\mathcal E_{ij}$ records evidence used by the feature; $\mathcal G_{ij}$ records the local causal role claim relative to the target trial; and $\mathcal M_{ij}$ records unresolved information, such as tacit clinical judgement not present in the EHR.
\end{definition}

The certificate may be produced by an AI system, a human reviewer, or a combined workflow. The mathematical requirement is not human review as such. It is that the proposed causal role be supported by evidence sufficient to place the feature in the locked target trial. A capable AI system can provide that evidence. If the relevant evidence is absent from the EHR and from the protocol, neither an AI system nor a human analyst can identify the role from the analysis table alone.

Let the possible roles be
\[
\mathcal T=\{B,D,H,U,M,Y,\mathsf{Obs},I,\mathsf{Excluded}\},
\]
where $B$ is baseline clinical state, $D$ is design or decision information, $H$ is prior history, $U$ is a candidate baseline marker or effect modifier, $M$ is a post treatment mediator, $Y$ is outcome or outcome proxy, $\mathsf{Obs}$ is observation or missingness process, $I$ is an intercurrent event, and $\mathsf{Excluded}$ denotes exclusion for the locked analysis.

For a locked target trial $L_T$, define
\[
\operatorname{Cert}_t(K_{ij},L_T)=1
\]
when the certificate passes the required predicates for role $t$. We write this as
\[
\operatorname{Cert}_t(K_{ij},L_T)=
I\{\mathsf{Time}_t=1,\ \mathsf{Source}_t=1,\ \mathsf{Avail}_t=1,\ \mathsf{Role}_t=1,\ \mathsf{Version}=1,\ \mathsf{Stability}=1\}.
\]
The predicates are role specific. $\mathsf{Time}_t$ checks the clinical and measurement clocks. $\mathsf{Source}_t$ checks the raw input window and whether source separation is required. $\mathsf{Avail}_t$ checks whether the information was available at the relevant decision time when the role is a decision or design role. $\mathsf{Role}_t$ checks the local causal claim. $\mathsf{Version}$ requires a locked representation version. $\mathsf{Stability}$ requires deterministic extraction, a fixed seed, or a declared seed distribution.

\begin{table}[H]
\centering
\caption{Operational rubric for causal certificates. A feature can enter the primary causal estimator only through a role whose predicates pass under the locked target trial.}
\label{tab:certificate_rubric}
\scriptsize
\resizebox{\textwidth}{!}{%
\begin{tabular}{>{\raggedright\arraybackslash}p{0.12\textwidth}>{\raggedright\arraybackslash}p{0.23\textwidth}>{\raggedright\arraybackslash}p{0.22\textwidth}>{\raggedright\arraybackslash}p{0.22\textwidth}>{\raggedright\arraybackslash}p{0.18\textwidth}}
\toprule
Role & Required evidence & Availability requirement & Failure condition & Permitted use \\
\midrule
$B$ baseline state & Clinical content and measurement refer to time before $T_i$; no post treatment component unless source separated & Not always required for biological state, but must be pre treatment content & Discharge summary, outcome adjacent note, mixed window without separation & Adjustment or standardization \\
$D$ design or decision & Eligibility, site, treatment availability, clinician decision information, monitoring rule, or target construction role is documented & Must be available at or before the treatment decision time & Replaced by a predictive score without safe compression & Retain, model, or stratify \\
$H$ prior history & Prior treatment, prior outcome, prior event, or prior care trajectory before $T_i$ & Required if used to reconstruct treatment decision; otherwise pre treatment history suffices & History reconstructed from post outcome narrative without separation & Adjustment, history summary \\
$U$ candidate marker & Pre treatment feature with stable extraction and prespecified or secondary status & Same as baseline state unless claimed as decision information & Outcome guided discovery used as confirmatory effect modifier & Adjustment, effect modification, secondary claim \\
$M$ mediator & Generated after treatment or encodes treatment response & Not available at $T_i$ by definition & Entered as baseline adjustment for total effect & Separate mediation lock \\
$Y$ outcome proxy & Encodes outcome, outcome risk after follow up begins, discharge status, improvement summary, or endpoint definition & Not admissible for baseline state & Entered in primary adjustment set & Outcome definition, exclusion, exploratory only \\
$\mathsf{Obs}$ observation process & Visit frequency, testing frequency, note density, censoring, follow up, measurement process & Relevant to observation or missingness mechanism & Treated only as disease state when it also determines observation & Observation or missingness model \\
$I$ intercurrent event & Post index treatment change, discontinuation, rescue therapy, transfer, death before endpoint, or competing event & Occurs after $T_i$ or changes interpretation after $T_i$ & Entered as baseline adjustment without strategy & Intercurrent event strategy \\
$\mathsf{Excluded}$ & Role cannot be certified, window is unresolved, version is unstable, or role is outside the lock & Not applicable & Any confirmatory causal use & Exclude or exploratory queue \\
\bottomrule
\end{tabular}%
}
\end{table}

Two clock distinctions are essential. First, $t^{clin}<T_i$ means that the clinical fact existed before treatment. Second, $t^{avail}\le T_i$ means that the information was available to the decision process. A post index note may legitimately contain a pre treatment fact, but that fact cannot be used to reconstruct treatment decision information unless the certificate also supports availability at the decision time. In this paper, source separation can certify baseline content; availability certification is a stronger claim.

The certificate has consequences. If $\operatorname{Cert}_t(K_{ij},L_T)=0$ for all primary state roles, the feature cannot enter the primary causal estimator as pre treatment information. If several roles remain possible, the target is a compatible estimand set. If the feature is outcome tuned, the claim is exploratory unless the extraction, selection, and estimation are separated by design.

\subsection*{Certificate workflow}

The certificate is intended to be executed as a workflow, not only read as a table. For each candidate AI derived feature, the following steps are performed before outcome modeling or effect estimation.

\begin{enumerate}
\item Lock the target trial: treatment versions, time zero, treatment decision milestone, outcome window, target population, observation strategy, and claim status.
\item Store the raw source pointer for the feature: note identifier, section, sentence span, order, laboratory record, imaging report, or event stream segment.
\item Extract four clocks: referenced clinical time, measurement time when applicable, EHR recording time, and availability time relative to the treatment decision.
\item Store the representation version: model, prompt or dictionary, preprocessing, schema, seed rule, post processing, and software hash when available.
\item Propose one or more causal roles in $\mathcal T$ and attach role evidence. A role without evidence remains only in the compatible role set; it is not accepted for the primary estimator.
\item Apply the role predicates in Table~\ref{tab:certificate_rubric}. A primary state or decision role is accepted only if the corresponding time, source, availability, role, version, and stability predicates pass.
\item If role predicates conflict, choose the most downstream role for the confirmatory point estimator, or report the compatible estimand set. Outcome proxy and mixed window flags dominate baseline labels unless source separation is certified.
\item Record all failures. A failure is not deletion of scientific information; it is a route to sensitivity analysis, mediation analysis, observation modeling, exploratory discovery, or next study measurement.
\end{enumerate}

This workflow prevents a common failure mode. A feature can be clinically meaningful and still fail as baseline information if its source is mixed, its availability is unknown, or its extraction was tuned after seeing the outcome. This is a routing decision, not deletion of information. Features that fail the point-estimator predicates are carried to compatible reporting, observation modeling, mediation locks, mechanism analyses, or next-study hypotheses. Conversely, a feature need not be discarded merely because it appears in a late note: if the pre treatment content is source separated and the intended role is baseline clinical state rather than treatment decision information, the certificate can still support a restricted use.

\subsection*{Architecture neutral extraction schema}

The framework does not require a specific NLP architecture. A rule based system, ClinicalBERT style classifier, retrieval augmented large language model, or human AI workflow can be used if it produces the evidence required by the certificate. The recommended extractor output is a structured record rather than a single feature value.

\begin{table}[H]
\centering
\caption{Minimum extraction schema for time aware certificate generation. The fields can be produced by an AI system, human reviewer, or combined workflow.}
\label{tab:extraction_schema}
\scriptsize
\resizebox{\textwidth}{!}{%
\begin{tabular}{>{\raggedright\arraybackslash}p{0.22\textwidth}>{\raggedright\arraybackslash}p{0.37\textwidth}>{\raggedright\arraybackslash}p{0.33\textwidth}}
\toprule
Field & Required content & Why it matters \\
\midrule
Feature value & Extracted value, category, score, or text span & Defines the candidate reconstructed information \\
Source pointer & Document, section, sentence span, event identifier, or order identifier & Permits source sentence audit and reproducibility \\
Recording time & Time the source entered the EHR & Detects post index or outcome adjacent sources \\
Referenced clinical time & Time the statement refers to, such as before admission or after treatment & Separates clinical facts from the time they were recorded \\
Measurement time & Time a test, assessment, or observation occurred, if applicable & Prevents post treatment measurements from being treated as baseline \\
Availability claim & Whether the information was available at the treatment decision time and the evidence for that claim & Distinguishes baseline clinical state from decision information \\
Treatment relation & Pre, post, mixed, or uncertain relative to the locked milestone & Routes features away from primary baseline adjustment when needed \\
Outcome relation & Independent, outcome adjacent, outcome proxy, or uncertain & Prevents endpoint information from entering the state object \\
Proposed role and uncertainty & Candidate role in $\mathcal T$, confidence, alternative roles, and unresolved facts & Defines the compatible role set and later reporting radius \\
\bottomrule
\end{tabular}%
}
\end{table}

The schema is deliberately stricter than ordinary feature extraction. A model that outputs only a frailty score or performance status label is not enough for confirmatory causal estimation. It must also expose the source and timing evidence that makes the label admissible for the locked causal role.

\section{Prospective approximation as comparison of information experiments}

Information set emulation can be stated without first naming any estimator. In the language of statistical experiment comparison \citep{lecam1986}, the prospective study and the AI typed EHR study are two statistical experiments. The question is whether the second can reproduce the information in the first closely enough for the target causal decision.

Let $X_P$ denote the pre treatment information that would be observed under the prospective study protocol. Its law under causal world $\eta\in\Omega$ is denoted $P^P_\eta$. Let $X_\Lambda$ denote the information observed after the AI typed lift: the treatment episode, typed ledger, point information set, compatible role set, and any locked design variables. Its law is denoted $P^\Lambda_\eta$. The prospective information experiment and the AI typed EHR information experiment are
\[
\mathcal E_P=\{P^P_\eta:\eta\in\Omega\},\qquad
\mathcal E_\Lambda=\{P^\Lambda_\eta:\eta\in\Omega\}.
\]
Here $\Omega$ indexes the causal worlds under consideration, including potential outcome laws, treatment choice or assignment mechanisms, EHR recording and observation mechanisms, and the distribution of the AI typed lift if it is stochastic before locking.

A randomized reconstruction kernel $Q$ maps the AI typed EHR observation $X_\Lambda$ to a pseudo prospective observation. Define the one sided prospective information deficit
\[
\Delta(\mathcal E_P,\mathcal E_\Lambda)
=
\inf_Q\sup_{\eta\in\Omega}
\left\|QP^\Lambda_\eta-P^P_\eta\right\|_{\mathrm{TV}}.
\]
The direction is intentional. If $\Delta(\mathcal E_P,\mathcal E_\Lambda)$ is small, then the AI typed EHR experiment can reproduce the prospective pre treatment information experiment up to a small error. If it is large, the retrospective study may still be useful, but it is not close to the prospective study at the level of information available before treatment.

\begin{theorem}[Risk transfer under information experiment approximation]
Suppose a decision rule $d_P$ in the prospective information experiment has loss $0\le \ell\{d_P(X_P),\eta\}\le L$. If $\Delta(\mathcal E_P,\mathcal E_\Lambda)\le \varepsilon$, then for every $\delta>0$ there exists a randomized EHR decision rule $d_{\Lambda,\delta}$ such that
\[
\sup_{\eta\in\Omega}
\left|E_\eta\ell\{d_{\Lambda,\delta}(X_\Lambda),\eta\}-E_\eta\ell\{d_P(X_P),\eta\}\right|
\le L(\varepsilon+\delta).
\]
Consequently, for the corresponding minimax risks,
\[
R_\Lambda^\star\le R_P^\star+L\varepsilon.
\]
\end{theorem}

\begin{proof}
Fix $\delta>0$. By the definition of the infimum and the assumption $\Delta(\mathcal E_P,\mathcal E_\Lambda)\le\varepsilon$, choose a kernel $Q_\delta$ such that
\[
\sup_\eta \left\|Q_\delta P^\Lambda_\eta-P^P_\eta\right\|_{\mathrm{TV}}\le \varepsilon+\delta.
\]
Given $X_\Lambda$, draw $\widetilde X_P\sim Q_\delta(\cdot\mid X_\Lambda)$ and apply $d_P$ to $\widetilde X_P$. This defines $d_{\Lambda,\delta}$. For each $\eta$, the distribution of $\widetilde X_P$ is $Q_\delta P^\Lambda_\eta$, while the prospective rule uses $P^P_\eta$. Since the loss is bounded by $L$, the difference in risks is at most $L\|Q_\delta P^\Lambda_\eta-P^P_\eta\|_{\mathrm{TV}}$, proving the first statement. For the minimax statement, apply the construction to a prospective rule within $\gamma>0$ of $R_P^\star$, take the infimum over retrospective rules, and then let $\delta\downarrow0$ and $\gamma\downarrow0$. If an optimal reconstruction kernel exists, one may set $\delta=0$ in the first statement.
\end{proof}

The theorem gives a formal meaning to the phrase ``retrospective like prospective.'' It does not say that an EHR study is prospective. It says that if the AI typed EHR experiment is close to the prospective information experiment, then any causal decision rule or estimator that would have been available prospectively can be approximated from the typed EHR information with only a controlled increase in risk.

This experiment distance is a global object. It is usually too ambitious to estimate directly in a real EHR study, and it is not the statistic routinely reported in an applied analysis. It should be read as a benchmark for what information set emulation means. The compatible estimand set below is its target specific counterpart. The experiment distance asks whether the whole typed EHR information experiment approximates the prospective information experiment. The information radius asks whether the particular causal estimand of interest has been narrowed enough for the locked claim. Thus $\Delta(\mathcal E_P,\mathcal E_\Lambda)$ is the broad design ideal, while $r(K)$ is the estimand specific diagnostic used for reporting and simulation. The two objects are not tied by a general KL or $f$-divergence lower bound in this paper; such a bound would require additional global smoothness or information projection assumptions on the estimand map. Operationally, the way to reduce the ideal gap is not to estimate $\Delta$ directly, but to refine $K$ through source pointer audit, clock parsing, availability adjudication, source separation, observation process modeling, and the Phase~0 validation described later.

The value of AI typing can also be expressed on the experiment scale. If $\mathcal E_0$ is the untyped retrospective experiment and $\mathcal E_\Lambda$ is the AI typed experiment, define
\[
\operatorname{VoT}_\Delta(\Lambda_\theta)
=
\Delta(\mathcal E_P,\mathcal E_0)-\Delta(\mathcal E_P,\mathcal E_\Lambda).
\]
This is the amount by which the AI typed lift moves the retrospective information experiment toward the prospective information experiment. The risk scale value based on the information radius, defined later, is the same idea restricted to a single target estimand.

The experiment distance is deliberately global. A smaller, target specific object is often more useful for analysis. If $K_P$ denotes the typed evidence that would be available in the prospective study for the same locked target, and $K$ denotes the typed evidence available in the retrospective EHR study, define the prospective emulation gap for the scalar estimand $\Psi$ as
\[
G_\Psi(K;K_P)=r^2(K)-r^2(K_P),
\]
whenever the two radii are evaluated on comparable reporting scales. In an ideal prospective information structure for which the locked estimand is point identified, $r(K_P)=0$, and the gap is simply $r^2(K)$. Thus the global distance $\Delta(\mathcal E_P,\mathcal E_\Lambda)$ asks whether the retrospective information experiment approximates the prospective information experiment as a whole, while $G_\Psi$ asks how much exact minimax risk remains for the particular causal target.

Published trial reports and protocols can help construct the prospective reference $K_P$. They do not determine the EHR information radius by themselves, because $r(K)$ also depends on the source, clock, availability, role evidence, validation-calibrated sensitivity envelopes, and observed law induced by the retrospective record. Their value is to anchor the target information structure: baseline measurements, stratification factors, eligibility criteria, observation schedules, intercurrent event strategies, and outcome definitions that the EHR typed lift must attempt to reproduce. In this limited sense, AI extraction from trial reports can serve as a reverse trial audit: it builds a structured checklist for Phase~0 validation, not a numerical substitute for the EHR certificate.

\section{Observational fibers and compatible estimands}

The relevant object is the set of causal worlds that remain possible after the EHR table and the typed evidence have been fixed. This section writes that set as an observational fiber and shows that its image under the target estimand determines the exact minimax risk of the retrospective information structure.

Let $\Theta$ denote a declared class of causal worlds. A world $\theta\in\Theta$ contains the prospective information law, treatment choice or assignment mechanism, observation mechanism, potential outcomes, EHR recording process, and the distribution of the AI typed lift once the version and stochastic rule are declared. The class $\Theta$ is not an unrestricted universe. It is the identification and sensitivity model named by the analyst: the target trial lock, exchangeability and observation assumptions, role restrictions, source separation constraints, availability error bounds, extraction uncertainty, and positivity restrictions. The information radius is therefore a conditional diagnostic. It is conditional on the observed law, the typed evidence, and the declared model $\Theta$. If a richer model $\Theta_2$ contains a more restrictive model $\Theta_1$, then for the same $p$ and $K$, $\mathfrak P_{\Theta_1,K}(p)\subseteq\mathfrak P_{\Theta_2,K}(p)$ and $r_{\Theta_1,K}(p)\le r_{\Theta_2,K}(p)$. The radius makes these assumptions visible rather than removing them. Given typed evidence $K$, let
\[
\Pi_K:\Theta\to\mathcal P(\mathcal O_K)
\]
be the observation map that sends a causal world to the joint law of all quantities available to the retrospective analyst under evidence $K$: the flat EHR table, the typed ledger, declared source windows, role evidence, availability evidence, and any locked validation information. For an observed law $p$ define the observational fiber
\[
\mathcal F_K(p)=\{\theta\in\Theta:\Pi_K(\theta)=p\}.
\]
A coarse flat table corresponds to a coarse observation map and therefore a large fiber. A stronger AI typed lift adds source time, availability, version, and causal role evidence, thereby replacing $K$ by a refinement and shrinking the fiber.

This refinement can be expressed as an order on evidence structures. We write $K_2\succeq K_1$ when the observation generated under $K_2$ can reproduce the observation generated under $K_1$, or equivalently when there is a measurable degradation map $R$ such that $\Pi_{K_1}=R\circ\Pi_{K_2}$. In that case $K_2$ is at least as informative as $K_1$. For an observed law $p_2$ under $K_2$, let $p_1=R(p_2)$ be its induced law under $K_1$. Every $K_2$ fiber at $p_2$ is then contained in the corresponding $K_1$ fiber at $p_1$. AI typing, source sentence audit, availability adjudication, and successful source separation are useful precisely when they move the retrospective study upward in this order.

For a locked scalar prospective estimand $\Psi:\Theta\to\mathbb R$, the compatible estimand set is the image of the fiber:
\[
\mathfrak P_K(p)=\Psi\{\mathcal F_K(p)\}
=\{\Psi(\theta):\theta\in\mathcal F_K(p)\}.
\]
When the observed law is clear we write $\mathfrak P(K)$. The information radius is the Chebyshev radius of this image, with the center taken over the real line:
\[
r(K)=\inf_{c\in\mathbb{R}}\sup_{\psi\in\mathfrak P(K)}|\psi-c|.
\]
If $\mathfrak P(K)$ is a closed bounded interval, $r(K)$ is half its width. If $r(K)=0$, the typed evidence and assumptions identify a single prospective estimand. If $r(K)$ is large, the retrospective study has not reconstructed the prospective information set tightly enough for a single target claim, regardless of sample size.

The point estimation set defined earlier is now easy to interpret. A feature can enter $\mathcal H_{\theta,K}^{point}$ only when its remaining role set is nonempty and contained in $\mathcal A_{pre}$, at least one pre-treatment role is certified, and the complete representation is pre-treatment measurable. A feature with both pre treatment and downstream compatible roles enlarges the fiber and therefore belongs in $\mathfrak P(K)$ rather than in the point estimator. This is not merely a conservative convention; it follows from the fiber still containing worlds in which that feature changes the target.

\begin{theorem}[Exact fiber radius minimax duality for retrospective information structures]
Fix typed evidence $K$ and an observed law $p$. Suppose $\mathfrak P_K(p)$ is nonempty and compact. Consider estimators measurable with respect to the observed data generated under $p$ and the typed evidence $K$. Then
\[
\inf_{\widehat\Psi}\sup_{\theta\in\mathcal F_K(p)}
E_\theta\{\widehat\Psi-\Psi(\theta)\}^2
=
r^2(K).
\]
A Chebyshev center of $\mathfrak P_K(p)$ attains the upper bound.
\end{theorem}

\begin{proof}
Every world in $\mathcal F_K(p)$ induces the same distribution for the data available to the estimator. Let $X$ be the random output of any estimator under this common law, and let $\mu_X=E(X)$. For any $\psi\in\mathfrak P_K(p)$,
\[
E(X-\psi)^2=\operatorname{Var}(X)+(\mu_X-\psi)^2.
\]
Therefore
\[
\sup_{\psi\in\mathfrak P_K(p)}E(X-\psi)^2
\ge
\sup_{\psi\in\mathfrak P_K(p)}(\mu_X-\psi)^2
\ge r^2(K),
\]
because $r(K)$ is the smallest possible worst case distance from a real center to $\mathfrak P_K(p)$. This proves the lower bound. For the upper bound, choose a real Chebyshev center $c_K$ of $\mathfrak P_K(p)$, so that $\sup_{\psi\in\mathfrak P_K(p)}|\psi-c_K|=r(K)$, and use the constant estimator $\widehat\Psi\equiv c_K$. Its worst case squared error over the fiber is exactly $r^2(K)$. Compactness guarantees the center exists.
\end{proof}

The theorem gives the radius its statistical meaning. It is not only a sensitivity width. Its square is the exact minimax mean squared error left by the information structure for the locked scalar estimand. The radius is not a sampling variance term. Its square is the residual information-structure risk that remains when the typed evidence and the observed law still do not distinguish among compatible causal worlds. More observations from the same typed EHR law can reduce sampling error around an identified functional, but they cannot remove fiber uncertainty. Only refinement of the typed evidence, such as source time, availability, or causal role evidence, can shrink that part of the risk. In applications the fiber and its Chebyshev center are not known objects; they are approximated by role enumeration, validation-calibrated sensitivity envelopes, and sensitivity analyses. The theorem states what those reported sets are trying to estimate: the residual risk created by the information structure itself.

The exact theorem treats causal worlds that induce the same typed retrospective law. Finite sample analysis also faces worlds that are not identical but are nearly indistinguishable at the available sample size. For a fixed sample size $n$, let $P_{\theta,K}^{(n)}$ denote the joint law of the $n$ observations available under typed evidence $K$ in world $\theta$. For $0\le\varepsilon<1$, define the approximate information modulus
\[
\omega_{K,n}(\varepsilon)=
\sup\left\{
|\Psi(\theta)-\Psi(\theta')|:
 d_{\mathrm{TV}}\bigl(P_{\theta,K}^{(n)},P_{\theta',K}^{(n)}\bigr)\le \varepsilon,
\ \theta,\theta'\in\Theta
\right\}.
\]
This quantity asks how far the locked estimand can move between causal worlds that the typed retrospective experiment cannot reliably distinguish with $n$ observations. It is not a replacement for $r(K)$. The radius is the sharp exact-fiber quantity; the modulus is its finite sample near-fiber analogue. It also should not be read as a general theorem linking the global experiment distance $\Delta(\mathcal E_P,\mathcal E_\Lambda)$ to the target-specific radius. No global information projection or smoothness assumption is being imposed here; the modulus works directly with near-indistinguishability under the typed retrospective law.

\begin{proposition}[Approximate fiber modulus lower bound]
Fix $K$, $n$, and $0\le\varepsilon<1$. For every $\eta>0$ and every estimator measurable with respect to the typed retrospective data,
\[
\sup_{\theta\in\Theta}
E_\theta\{\widehat\Psi-\Psi(\theta)\}^2
\ge
\frac{1-\varepsilon}{8}\,\{\omega_{K,n}(\varepsilon)-\eta\}_+^2.
\]
Consequently, when $\omega_{K,n}(\varepsilon)<\infty$,
\[
\inf_{\widehat\Psi}\sup_{\theta\in\Theta}
E_\theta\{\widehat\Psi-\Psi(\theta)\}^2
\ge
\frac{1-\varepsilon}{8}\,\omega_{K,n}^2(\varepsilon).
\]
\end{proposition}

\begin{proof}
If $\omega_{K,n}(\varepsilon)=0$, the claim is immediate. Otherwise fix $\eta>0$ and choose $\theta,\theta'\in\Theta$ such that $d_{\mathrm{TV}}(P_{\theta,K}^{(n)},P_{\theta',K}^{(n)})\le\varepsilon$ and $\delta=|\Psi(\theta)-\Psi(\theta')|\ge\{\omega_{K,n}(\varepsilon)-\eta\}_+$. Let $A$ be the event that an estimator is closer to $\Psi(\theta)$ than to $\Psi(\theta')$, with ties assigned arbitrarily. On $A^c$ under $\theta$, the squared error is at least $\delta^2/4$; on $A$ under $\theta'$, the squared error is at least $\delta^2/4$. Hence the larger of the two risks is at least
\[
\frac{\delta^2}{4}\max\{P_{\theta,K}^{(n)}(A^c),P_{\theta',K}^{(n)}(A)\}.
\]
Moreover,
\[
P_{\theta,K}^{(n)}(A^c)+P_{\theta',K}^{(n)}(A)
=1-P_{\theta,K}^{(n)}(A)+P_{\theta',K}^{(n)}(A)
\ge 1-d_{\mathrm{TV}}(P_{\theta,K}^{(n)},P_{\theta',K}^{(n)})
\ge 1-\varepsilon.
\]
The maximum of the two probabilities is therefore at least $(1-\varepsilon)/2$, and the two-point lower bound is at least $(1-\varepsilon)\delta^2/8$. The displayed result follows from the choice of the pair and then from the arbitrariness of $\eta$.
\end{proof}

This lower bound is not an estimator construction and does not provide a finite sample confidence procedure. It identifies an unavoidable obstruction: if two compatible worlds are statistically close under the typed retrospective experiment while their locked estimands differ, no estimator can uniformly remove that ambiguity.

The proposition is the finite sample version of the same obstruction. If the typed evidence leaves two causal worlds nearly indistinguishable at the observed sample size, and those worlds imply different locked estimands, then no estimator can overcome that ambiguity. The exact radius theorem is sharper on a fixed fiber; the modulus lower bound explains how the same geometry persists when equality of laws is relaxed to near-indistinguishability. This is a local, target-specific statement about the typed retrospective experiment, not an attempt to turn the global experiment distance into an automatic formula for $r(K)$.

\begin{corollary}[Flat table nonidentification]\label{cor:flat_nonid}
Let $K_0$ contain only a flat table law for $(A,Z,Y)$ and no evidence about source time, availability, or causal role. There are laws for which $\mathcal F_{K_0}(p)$ contains a world in which $Z$ is pre treatment information, a world in which $Z$ is post treatment mediator information, and a world in which $Z$ is outcome adjacent proxy information. If the corresponding prospective estimands differ, then $r(K_0)>0$ and no estimator using only the flat table can identify the prospective causal effect.
\end{corollary}

\begin{proof}
Given a positive law $P(A,Z,Y)$, one compatible construction draws $Z\sim P(Z)$, $A\sim P(A\mid Z)$, and $Y\sim P(Y\mid A,Z)$; this interprets $Z$ as pre treatment information. A second construction draws $A\sim P(A)$, $Z\sim P(Z\mid A)$, and $Y\sim P(Y\mid A,Z)$; this interprets $Z$ as post treatment. A third draws $(A,Y)\sim P(A,Y)$ and then $Z\sim P(Z\mid A,Y)$; this interprets $Z$ as outcome adjacent. These worlds share the same flat table but can imply different prospective target estimands. Their images under $\Psi$ therefore have positive diameter whenever the estimands differ, and the exact radius theorem applies.
\end{proof}

\begin{proposition}[Monotonicity under typed evidence refinement]
If $K_2\succeq K_1$ with degradation map $R$ satisfying $\Pi_{K_1}=R\circ\Pi_{K_2}$, then for every observed law $p_2$ under $K_2$ and its induced law $p_1=R(p_2)$ under $K_1$,
\[
\mathcal F_{K_2}(p_2)\subseteq\mathcal F_{K_1}(p_1),
\]
and therefore
\[
\mathfrak P(K_2)\subseteq\mathfrak P(K_1),\qquad r(K_2)\le r(K_1),\qquad r^2(K_2)\le r^2(K_1).
\]
If $K_P$ is the prospective evidence structure and $K_2\succeq K_1$, then
\[
G_\Psi(K_2;K_P)\le G_\Psi(K_1;K_P)
\]
whenever both gaps are defined.
\end{proposition}

\begin{proof}
If $\theta\in\mathcal F_{K_2}(p_2)$, then $\Pi_{K_2}(\theta)=p_2$, so $\Pi_{K_1}(\theta)=R\{\Pi_{K_2}(\theta)\}=R(p_2)=p_1$. Hence $\theta\in\mathcal F_{K_1}(p_1)$. Taking the image of this inclusion under $\Psi$ preserves inclusion. The Chebyshev radius of a subset cannot exceed that of the original set. Squaring preserves the order because the radii are nonnegative. The statement about $G_\Psi$ follows by subtracting the same prospective radius $r^2(K_P)$ from both sides.
\end{proof}

The proposition is the formal role of AI typing and certification. A certificate does not create trust by assertion. It refines the evidence structure and removes causal worlds from the observational fiber. What remains is retained as a compatible estimand set.

\subsection*{Value of typing and value of certification}

Let $K_0$ denote the untyped state of the study: the flat EHR feature table with no usable evidence about source time, availability, or causal role. Let $K_\theta$ be the evidence produced by the AI typed lift $\Lambda_\theta$. For a finite reporting scale $M>0$, define $r_M(K)=\min\{r(K),M\}$. The radius scale value of AI typing is
\[
\operatorname{VoT}_M(\Lambda_\theta)=r_M(K_0)-r_M(K_\theta).
\]
The corresponding capped squared-radius reporting score is
\[
\operatorname{VoT}_{\mathrm{risk},M}(\Lambda_\theta)=r_M^2(K_0)-r_M^2(K_\theta).
\]
This equals the reduction in exact minimax mean squared error only when the cap is inactive at both evidence states. In general it is a capped reporting score, not the minimax risk reduction for the original squared loss. Define
\[
G_{\Psi,M}(K;K_P)=r_M^2(K)-r_M^2(K_P).
\]
For a fixed prospective reference, the corresponding change is
\[
\operatorname{VoT}_{\mathrm{gap},M}(\Lambda_\theta;K_P)
=G_{\Psi,M}(K_0;K_P)-G_{\Psi,M}(K_\theta;K_P),
\]
which equals the capped squared-radius score when the same prospective reference is used. If the lift is stochastic before locking, the design stage analogue is the expectation of this reduction under the declared distribution of the lift.

Additional evidence can be evaluated in the same geometry. If $E$ is a source sentence audit, availability adjudication, chart review subset, repeated extraction audit, external validation sample, or prospective substudy, write $K\vee E$ for the refined evidence. The value of certification is
\[
\operatorname{VoC}_M(E\mid K)=r_M(K)-r_M(K\vee E),
\]
with capped squared-radius reporting version
\[
\operatorname{VoC}_{\mathrm{risk},M}(E\mid K)=r_M^2(K)-r_M^2(K\vee E).
\]
When a cost $C(E)$ is available, the next validation action can be prioritized by
\[
E^\star\in\arg\max_E \frac{\operatorname{VoC}_M(E\mid K)}{C(E)},
\]
or by replacing the numerator with the capped squared-radius score. The certification score is an exact minimax risk reduction only when its cap is inactive at both states. A wide radius is therefore not only a stopping rule. It is a design diagnosis that says which evidence would shrink the fiber most.

\subsection*{Estimating a compatible reporting set}

The exact fiber radius is a population object. A compatible reporting set requires a declared identification, validation, and sensitivity model. Fix a finite collection of remaining compatible role assignments $\mathcal C(K)$, with $m=|\mathcal C(K)|\ge1$, before using the inference sample. Each assignment specifies a nominal functional $\Psi_t^0(p)$ for the same scientific estimand $\Psi$ under the same target-trial lock. A direct-effect target and a total-effect target must not be pooled as though they were alternative values of one locked estimand.

Choose error budgets $\alpha_S+\alpha_V\le\alpha$. Suppose the nominal estimates have valid marginal Wald intervals under each relevant world, uniformly as needed for the stated coverage claim. A Bonferroni choice is $c_S=z_{1-\alpha_S/(2m)}$; a justified simultaneous procedure may replace it. With estimand-scale envelopes $\widehat B_t(K)$, report
\[
\widehat{\mathfrak P}_{1-\alpha}(K)=
\bigcup_{t\in\mathcal C(K)}
\left[\widehat\Psi_t-\widehat B_t(K)-c_S\widehat{se}_t,
\widehat\Psi_t+\widehat B_t(K)+c_S\widehat{se}_t\right].
\]
If a scalar summary is required, use the center of its enclosing interval,
\[
\widehat c(K)=\frac{\inf\widehat{\mathfrak P}_{1-\alpha}(K)+\sup\widehat{\mathfrak P}_{1-\alpha}(K)}{2},\qquad
\widehat r_{1-\alpha}(K)=\frac{\sup\widehat{\mathfrak P}_{1-\alpha}(K)-\inf\widehat{\mathfrak P}_{1-\alpha}(K)}{2}.
\]
The union can be disconnected and need not contain its enclosing-interval center. The set remains the primary causal result. Its reporting radius combines sampling and information uncertainty; these components should also be displayed separately when they can be assessed.

A candidate validation or sensitivity envelope is
\[
\widehat B_t(K)=
\widehat B_{\mathrm{extract},t}+\widehat B_{\mathrm{avail},t}+
\widehat B_{\mathrm{sep},t}+\widehat B_{\mathrm{obs},t}+
\widehat B_{\mathrm{judge},t}+\widehat B_{\mathrm{pos},t}.
\]
A raw extraction error rate is not an estimand-scale bound. Coverage requires a map from the declared discrepancies to the locked estimand.

\begin{proposition}[Calibration of an estimand-scale validation envelope]
Fix $p$, a locked evidence rule $K$, and the finite assignment set $\mathcal C(K)$. Suppose an event $V$ has probability at least $1-\alpha_V$ and the following statements hold simultaneously on $V$ for every $\theta\in\mathcal F_K(p)$. Some retained assignment $t$ represents $\theta$, and there is a path from its nominal target $\Psi_t^0(p)$ to $\Psi(\theta)$ such that, for each component $j$,
\[
d_j(\theta)\le\delta_j,\qquad
|\Psi(\theta^{(j)})-\Psi(\theta^{(j-1)})|\le L_{jt}d_j(\theta).
\]
If $\widehat B_t(K)\ge\sum_j L_{jt}\delta_j$ on $V$ and all retained nominal-target sampling intervals cover simultaneously with probability at least $1-\alpha_S$, then
\[
P\{\mathfrak P_K(p)\subseteq\widehat{\mathfrak P}_{1-\alpha}(K)\}
\ge1-\alpha_S-\alpha_V.
\]
If the sampling guarantee is asymptotic, the conclusion has an additional $o(1)$ term. If the validation and representation conditions hold only for the actual world rather than every world in the fiber, the conclusion is point coverage of its target, not coverage of the entire compatible image.
\end{proposition}

\begin{proof}
On $V$, telescope along each world's declared path. The triangle inequality bounds its target's distance from the retained nominal value by $\sum_jL_{jt}\delta_j$. On the simultaneous sampling event, that nominal value lies in its sampling interval, so expanding by $\widehat B_t(K)$ covers the world's target. Applying this argument to every world gives image coverage. The two events fail with probability at most $\alpha_S+\alpha_V$, without requiring independence.
\end{proof}

For example, if $Y\in[y_{min},y_{max}]$ with range $R_Y$ and the arm-specific interventional outcome laws differ from their nominal laws by total variation distances at most $\delta_1$ and $\delta_0$, then the treatment-contrast envelope is $R_Y(\delta_1+\delta_0)$. The MNAR construction in Appendix~C is another explicit estimand-scale map. In contrast, an extraction error rate or availability disagreement rate alone does not determine $L_{jt}\delta_j$; without a transport or calibration model linking that rate to the relevant interventional law, the component must be labeled a sensitivity allowance and cannot be claimed to have coverage.

These envelopes are validation or sensitivity quantities, not hidden tuning parameters. They should be reported even when set to zero. Each component that is claimed as a bound must identify $d_j$, $\delta_j$, $L_{jt}$, and its validation event; otherwise it remains a declared sensitivity allowance. Reviewer disagreement is handled in the same geometry. If clinicians or epidemiologic reviewers disagree about a feature's source window, availability, or role, that disagreement widens the compatible set through $\widehat B_{\mathrm{judge},t}$ or through a larger role set; it is not treated as an ordinary sampling standard error and it cannot by itself justify entry into the point estimator. In discrete implementations, lower and upper compatible bounds can also be written as constrained bounding problems,
\[
\underline\Psi(K)=\inf_{\theta\in\Theta(P_O,K,\mathcal V)}\Psi(\theta),\qquad
\overline\Psi(K)=\sup_{\theta\in\Theta(P_O,K,\mathcal V)}\Psi(\theta),
\]
where $\mathcal V$ encodes validation constraints such as source separation error, availability error, extraction uncertainty, and positivity bounds. The radius is then $(\overline\Psi(K)-\underline\Psi(K))/2$. Depending on the estimand and discretization, this bounding problem may be linear, linear fractional, or nonlinear; the framework requires that the constraints be stated rather than hidden in a point estimate.

\begin{table}[H]
\centering
\caption{Components of a validation-calibrated sensitivity envelope for a compatible reporting set. A component is a coverage bound only when its discrepancy, validation limit, and estimand modulus are stated.}
\label{tab:validation_components}
\scriptsize
\resizebox{\textwidth}{!}{%
\begin{tabular}{>{\raggedright\arraybackslash}p{0.17\textwidth}>{\raggedright\arraybackslash}p{0.32\textwidth}>{\raggedright\arraybackslash}p{0.25\textwidth}>{\raggedright\arraybackslash}p{0.20\textwidth}}
\toprule
Component & Operational source & Example diagnostic & Effect on reporting \\
\midrule
$B_{\mathrm{extract}}$ & Manual annotation subset, repeated AI extraction, external NLP validation & feature error rate, source sentence precision, role disagreement & widens role specific intervals involving extracted features \\
$B_{\mathrm{avail}}$ & Adjudication of whether the fact was known at treatment decision time & proportion of features with pre clinical time but uncertain availability & separates baseline state use from decision information use \\
$B_{\mathrm{sep}}$ & Source separation audit on mixed or post index notes & fraction of outputs that change when post treatment sentences are masked & excludes or bounds features from late notes \\
$B_{\mathrm{obs}}$ & Visit, test, note, and follow up process sensitivity analysis & change in estimate under alternative observation models & records remaining observation process drift \\
$B_{\mathrm{judge}}$ & Chart review, clinician adjudication, or prospective substudy for tacit judgement & residual disagreement about performance status or treatment rationale & leaves clinician judgement as compatible uncertainty \\
$B_{\mathrm{pos}}$ & Overlap diagnostics after reconstruction & effective sample size, extreme weights, sparse strata & prevents precise estimates in unsupported information strata \\
\bottomrule
\end{tabular}%
}
\end{table}

Table~\ref{tab:radius_shrink} lists practical ways to shrink the radius. The list is also a stopping rule. If none of these actions is possible and the radius remains large, the analysis should not be reported as a point identified prospective approximation causal effect.

\begin{table}[H]
\centering
\caption{Practical routes for reducing the information radius.}
\label{tab:radius_shrink}
\resizebox{\textwidth}{!}{%
\begin{tabular}{>{\raggedright\arraybackslash}p{0.28\textwidth}>{\raggedright\arraybackslash}p{0.34\textwidth}>{\raggedright\arraybackslash}p{0.31\textwidth}}
\toprule
Uncertainty source & How to reduce it & What remains if not reduced \\
\midrule
Clinical time ambiguity & Time aware extraction, note section tagging, chart review of source sentence and event time & Competing baseline and post treatment roles \\
Availability ambiguity & Verify that the fact was known at decision time through prior note, order, structured field, or clinician review & Feature can reconstruct baseline state but not decision information \\
Outcome proxy risk & Exclude outcome adjacent sections, discharge summaries, and endpoint definitions from baseline extraction & Role set retains $Y$ and primary point identification fails \\
Observation process ambiguity & Record visit frequency, test frequency, note density, censoring, and follow up as observation variables & Drift through the observation mechanism remains \\
Extractor instability & Lock model, prompt, schema, seed, and preprocessing; repeat extraction to estimate instability & Measurement or generation error widens $\widehat B_t(K)$ \\
Unrecorded clinician judgement & Chart review, prospective substudy, external validation, or sensitivity parameter for residual judgement & Compatible set remains wide \\
\bottomrule
\end{tabular}%
}
\end{table}
\section{Causal effect identification and estimation}

Point estimation begins only after information set emulation has done enough work. The reconstructed point information set is useful because it supports a causal estimand and an estimator, but the estimator is not the primary innovation. The primary innovation is the typed lift that determines which features may enter $\mathcal H_{\theta,K}^{point}$, which features must remain in $J_{amb}(K)$, and whether the radius says a point estimator is scientifically adequate. Let $H_i$ denote a finite representation that generates $\mathcal H_{\theta,K}^{point}$ for patient $i$ in the EHR target frame. The target distribution $Q_K$ may be the empirical distribution of $H_i$ among eligible EHR patients, a transported distribution, or another distribution fixed in the target trial lock.

The typed lift does not make treatment assignment ignorable. It determines which information can be admitted into the pre treatment information set and which information must remain outside the point estimator. Exchangeability is a separate scientific and design assumption about that admitted information set. Thus the certificate prevents post treatment, outcome adjacent, or observation process features from being treated as baseline conditions; it does not assert that all residual clinician judgement, site practice, severity, or observation mechanisms have been measured. In short, the typed lift defines the information set relative to which exchangeability is assessed; it does not assert exchangeability.

The operational causal effect after information set emulation is
\[
\Psi_K=\int \{m_1^K(h)-m_0^K(h)\}\,dQ_K(h),
\]
where
\[
m_a^K(h)=E\{Y(a)\mid H=h\}.
\]
This estimand is the effect identified by the certified EHR information set. It coincides with the prospective study estimand only when the reconstructed information set is sufficiently close to the prospective information set for the scientific purpose at hand. The remaining distance is handled later through compatible estimands and the certificate conditional approximation frontier.

For point identification, assume the following conditions under the locked target trial. First, certified pre treatment measurability holds: $H$ is fixed before treatment and $H(1)=H(0)$. Second, consistency holds for the treatment versions being compared. Third, information set exchangeability holds:
\[
A\perp \{Y(1),Y(0)\}\mid H.
\]
Fourth, when outcomes can be missing, observation exchangeability holds:
\[
\Delta\perp \{Y(1),Y(0)\}\mid A,H,
\]
where $\Delta=1$ indicates that the outcome is observed. Fifth, positivity holds for treatment and observation on the support of $Q_K$:
\[
0<P(A=a\mid H=h)<1,
\qquad
P(\Delta=1\mid A=a,H=h)>0.
\]

Observation exchangeability is often the least credible condition in EHR studies. If outcome observation may depend on unrecorded severity, clinician judgement, or workflow intensity, the point identified analysis should be replaced by an observation sensitivity family rather than treated as verified. Let $U_{obs}$ denote an unmeasured observation driver. Define the selection ratio
\[
L_a^\theta(H,U_{obs})=
\frac{P_\theta(\Delta=1\mid A=a,H,U_{obs})}
{P_\theta(\Delta=1\mid A=a,H)},
\]
which measures the multiplicative dependence of outcome observation on the unmeasured driver after conditioning on the admitted information set. Let $\Theta_{obs}(\Gamma)$ be the declared class in which this dependence is bounded,
\[
\exp(-\Gamma_a)\le L_a^\theta(H,U_{obs})\le \exp(\Gamma_a),\qquad a=0,1.
\]
The corresponding compatible set and radius are
\[
\mathfrak P_{\Gamma}(K)=\Psi\{\mathcal F_{K,\Gamma}(p)\},
\qquad
r_{\Gamma}(K)=\inf_{c\in\mathbb R}
\sup_{\psi\in\mathfrak P_{\Gamma}(K)}|\psi-c|.
\]
If $\Gamma_1\le\Gamma_2$ componentwise, then the sensitivity class expands, $\mathfrak P_{\Gamma_1}(K)\subseteq\mathfrak P_{\Gamma_2}(K)$, and $r_{\Gamma_1}(K)\le r_{\Gamma_2}(K)$. Thus departures from observation exchangeability are reported as an information-radius curve, not hidden inside a single complete-case point estimate.

Define
\[
\pi_a(h)=P(A=a\mid H=h),
\]
\[
\rho_a(h)=P(\Delta=1\mid A=a,H=h),
\]
\[
\mu_a(h)=E(Y\mid A=a,\Delta=1,H=h).
\]

\begin{theorem}[Identification after information set emulation]
Under certified pre treatment measurability, consistency, information set exchangeability, observation exchangeability, and positivity,
\[
\Psi_K=\int\{\mu_1(h)-\mu_0(h)\}\,dQ_K(h).
\]
\end{theorem}

\begin{proof}
By consistency, among patients with $A=a$ and observed outcomes, $Y=Y(a)$. Observation exchangeability gives
\[
E(Y\mid A=a,\Delta=1,H=h)=E\{Y(a)\mid A=a,H=h\}.
\]
Information set exchangeability then gives
\[
E\{Y(a)\mid A=a,H=h\}=E\{Y(a)\mid H=h\}=m_a^K(h).
\]
Substituting the equality for both treatment levels into the definition of $\Psi_K$ gives the result.
\end{proof}

For estimation, first consider one episode per patient and independent, identically distributed patients from the eligible EHR frame. In this paragraph all expectations refer to that frame. Write $\tau_K(h)=\mu_1(h)-\mu_0(h)$. Distinguish the population target $\Psi_K=E\{\tau_K(H)\}$ from the empirical-covariate target $\Psi_{K,n}=n_E^{-1}\sum_i\tau_K(H_i)$. The following theorem concerns these two targets; a transported $Q_K$ requires its own weighting and target-distribution influence component.

The augmented inverse probability weighted contribution is
\[
\phi_i^K=
\frac{1(A_i=1)\Delta_i}{\pi_1(H_i)\rho_1(H_i)}\{Y_i-\mu_1(H_i)\}
-\frac{1(A_i=0)\Delta_i}{\pi_0(H_i)\rho_0(H_i)}\{Y_i-\mu_0(H_i)\}
+\tau_K(H_i).
\]
Define the residual contribution $R_i^K=\phi_i^K-\tau_K(H_i)$. The estimator is
\[
\widehat\Psi_K=n_E^{-1}\sum_{i=1}^{n_E}\widehat\phi_i^K,
\]
where every fitted nuisance value is evaluated out of fold. The estimator can be used for either target, but its first-order variance depends on which target is reported. This is an application of orthogonal estimation \citep{dml2018}.

\begin{theorem}[Orthogonal estimation after information set emulation]\label{thm:orthogonal}
Assume the identification conditions above and an externally fixed information-set rule, so that the evaluation patients are i.i.d. conditional on that rule. Assume $E\{\tau_K(H)^2\}<\infty$, uniformly bounded conditional residual second moments $E[\{Y-\mu_a(H)\}^2\mid A=a,\Delta=1,H]$, and true and fitted treatment and observation probabilities bounded below by a fixed positive constant. Use a fixed number of cross-fitting folds whose sizes are proportional to $n_E$. Let $\|\cdot\|$ be the $L_2$ norm under the eligible-frame law of $H$. For each arm and each training fold, require
\[
\|\widehat\mu_a-\mu_a\|+\|\widehat\pi_a-\pi_a\|+
\|\widehat\rho_a-\rho_a\|=o_p(1),
\]
and
\[
\|\widehat\mu_a-\mu_a\|
\bigl(\|\widehat\pi_a-\pi_a\|+\|\widehat\rho_a-\rho_a\|\bigr)
=o_p(n_E^{-1/2}).
\]
Then
\[
\sqrt{n_E}(\widehat\Psi_K-\Psi_K)
=n_E^{-1/2}\sum_i(\phi_i^K-\Psi_K)+o_p(1),
\]
whereas
\[
\sqrt{n_E}(\widehat\Psi_K-\Psi_{K,n})
=n_E^{-1/2}\sum_iR_i^K+o_p(1).
\]
For positive limiting variances, the limiting laws are normal with variances
$V_{pop}=\operatorname{Var}(\phi_i^K)$ and $V_{emp}=E\{(R_i^K)^2\}$, respectively. In particular,
\[
V_{pop}=V_{emp}+\operatorname{Var}\{\tau_K(H)\}.
\]
With $\widehat R_i^K=\widehat\phi_i^K-
\{\widehat\mu_1(H_i)-\widehat\mu_0(H_i)\}$ and $\overline{\widehat R}=n_E^{-1}\sum_i\widehat R_i^K$, consistent variance estimates are
\[
\widehat V_{pop}=n_E^{-1}\sum_i(\widehat\phi_i^K-\widehat\Psi_K)^2,
\qquad
\widehat V_{emp}=n_E^{-1}\sum_i(\widehat R_i^K-\overline{\widehat R})^2.
\]
The corresponding standard error is $\sqrt{\widehat V/n_E}$. The empirical-target statement concerns the marginal error around the random covariate-average target; a conditional-on-covariates coverage statement additionally requires a conditional central limit theorem.
\end{theorem}

\begin{proof}
For each arm, set $q_a=\pi_a\rho_a$ and $\widehat q_a=\widehat\pi_a\widehat\rho_a$. Conditional on the training fold, the bias of its augmented score is
\[
E\left[\left(1-\frac{q_a}{\widehat q_a}\right)
(\widehat\mu_a-\mu_a)\right].
\]
Positivity, Cauchy--Schwarz, and
$\|\widehat q_a-q_a\|\le\|\widehat\pi_a-\pi_a\|+\|\widehat\rho_a-\rho_a\|$
give the stated product bound. Individual consistency and the residual-moment bound imply $L_2$ convergence of the fitted scores. Conditional variance bounds within each held-out fold then make the centered score-estimation error $o_p(n_E^{-1/2})$; a fixed number of folds preserves this order. The i.i.d. central limit theorem gives the population expansion. Subtracting $\Psi_{K,n}-\Psi_K$ gives the residual expansion. Since $E(R_i^K\mid H_i)=0$, the variance decomposition follows. Score convergence and the law of large numbers give consistency of the two empirical variance estimates.
\end{proof}

The product-rate condition alone does not imply the displayed influence function: individual consistency (or a direct fitted-score convergence condition with an appropriate limit) is also needed. With repeated episodes, cross-fitting must split patients and the score must be aggregated using the declared within-patient target weights. A cluster extension requires $L_2$ convergence and the product remainder for those aggregate scores, a cluster central limit theorem, and a consistent variance estimate across independent patients. Arbitrary dependent episodes cannot be treated as independent contributions. Longitudinal treatment and observation mechanisms require their corresponding sequential identification assumptions \citep{robins2000}.

\subsection*{AI extraction regimes}

The theorem conditions on a locked information set $H$. That is appropriate only when the extraction rule is fixed before causal effect estimation. We distinguish three regimes.

\begin{table}[H]
\centering
\caption{AI extraction regimes and their consequences for causal estimation.}
\label{tab:extractor_regimes}
\resizebox{\textwidth}{!}{%
\begin{tabular}{>{\raggedright\arraybackslash}p{0.22\textwidth}>{\raggedright\arraybackslash}p{0.36\textwidth}>{\raggedright\arraybackslash}p{0.34\textwidth}}
\toprule
Regime & Definition & Consequence \\
\midrule
Locked external extractor & The feature extractor, prompt, dictionary, model version, seed rule, and schema are trained or chosen outside the outcome analysis and locked before estimation & Treat $H$ as an observed reconstructed information set, then use the identification and estimation results above \\
Cross fitted study extractor & The extractor is learned inside the study using pre outcome information, with extraction folds separated from causal estimation folds at patient or site level & Include extraction in the cross fitting plan and control representation error relative to a fixed target; Theorem~\ref{thm:orthogonal} alone does not cover arbitrary learned extractors \\
Outcome tuned extractor & Outcome data, treatment effect estimates, interaction tests, or post hoc prompt changes influence the feature definition used in the same confirmatory analysis & Not confirmatory under this framework without sample splitting, selective inference, external validation, or a new lock \\
\bottomrule
\end{tabular}%
}
\end{table}

If the extractor is stochastic, the seed must be fixed, or the seed distribution must be declared before outcome analysis. Repeated extraction can be used to estimate a generation instability component, for example $\operatorname{Var}_{\xi}\{R_\theta(E_i,\xi)\}$. That component is not a causal type by itself; it enters the certificate as measurement instability and may widen the compatible estimand set or the validation-calibrated sensitivity envelope.

Two kinds of uncertainty are kept separate. Value uncertainty means that the feature value changes across seeds, folds, prompts, or extractor versions while its causal role remains the same. This uncertainty may be summarized by repeated extraction and added to the statistical variance or to $B_{\mathrm{extract}}$. Role uncertainty means that the feature could be baseline information, mediator, outcome proxy, observation process, or another role. Role uncertainty is not an ordinary standard error. It changes the estimand and therefore belongs in $\mathfrak P(K)$ through the sets $J_{amb}(K)$ and $\mathcal T_j(K)$. A false admissible rate,
\[
\widehat p_{FA}=\widehat{\Pr}\{T\notin\mathcal A_{pre}\mid\text{point admitted in the validation subset}\},
\]
should be reported when annotation or adjudication data are available. If this rate is not acceptably small for the locked claim, the analysis should not report a single confirmatory point estimator.

For study learned extractors, fold variation is not ignored. Let $\widehat\Psi^{(g)}_K$ denote the causal estimate obtained when fold $g$ is held out for extraction or estimation according to the prespecified cross fitting plan. A descriptive fold-variation width is
\[
B_{\mathrm{fold}}=\frac{1}{2}\left\{\max_g \widehat\Psi^{(g)}_K-\min_g \widehat\Psi^{(g)}_K\right\}.
\]
For stochastic extractors with repeated seeds $b=1,\ldots,B$, an analogous seed width is
\[
B_{\mathrm{seed}}=\frac{1}{2}\left\{\max_b \widehat\Psi^{(b)}_K-\min_b \widehat\Psi^{(b)}_K\right\}.
\]
These quantities are not substitutes for causal identification. They may be declared as sensitivity allowances in the extraction component, but their observed ranges do not by themselves provide calibrated coverage bounds. If outcome data influence prompts, dictionaries, model selection, feature choice, or role labels, the resulting extractor is outcome tuned and the effect estimate is not confirmatory unless an independent split, external validation set, or a justified selection-adjusted procedure is used \citep{lee2016}.

Thus the method proceeds to causal estimation. The certificate determines which information set $H$ is allowed to define the estimand and nuisance functions. Once that is fixed and the identification conditions are credible, standard semiparametric estimators can be used. If the certificate is not strong enough to fix a single role for the reconstructed information, the target becomes a compatible estimand set rather than a point identified effect.

\section{EHR compression drift}

The next result applies the conditional-covariance compression identity to the joint EHR observation mechanism. D'Amour and Franks \citep{damour2021}, Propositions 1--2 and Appendix B, characterize representation-induced bias through conditional covariance; Clivio et al. \citep{clivio2026}, Lemma 3.1, develop a density-ratio version for the average treatment effect on the treated. Here the tilting probability combines EHR frame presence, treatment assignment, and outcome observation. The application connects this established mechanism to the certificate architecture; the covariance identity and its zero-covariance criterion are not claimed as new general compression theory.

Let $W$ denote the prospective information needed for the target trial:
\[
W=(X,D,J,O),
\]
where $X$ is clinical state, $D$ is design and eligibility information, $J$ is clinician judgement and patient preference relevant to treatment, and $O$ is the observation process. Let $Z=T(W)$ be an AI reconstructed information set. Define
\[
s(W)=P(S=1\mid W),
\]
where $S=1$ indicates presence in the EHR analytic frame,
\[
\pi_a(W)=P(A=a\mid S=1,W),
\]
where $\pi_a$ is the treatment assignment or treatment choice mechanism, and
\[
\rho_a(W)=P(\Delta=1\mid S=1,A=a,W),
\]
where $\Delta=1$ indicates outcome observation. Let
\[
m_a(W)=E\{Y(a)\mid W\}.
\]

\begin{theorem}[EHR compression drift]
Assume consistency. Assume $S\perp Y(a)\mid W$, $A\perp Y(a)\mid S=1,W$, and $\Delta\perp Y(a)\mid S=1,A=a,W$. Let $q_a(W)=s(W)\pi_a(W)\rho_a(W)$, and assume $E\{q_a(W)\mid Z=z\}>0$. Then
\[
E(Y\mid S=1,A=a,\Delta=1,Z=z)-E\{Y(a)\mid Z=z\}
=
\frac{\operatorname{Cov}\{m_a(W),q_a(W)\mid Z=z\}}
{E\{q_a(W)\mid Z=z\}}.
\]
\end{theorem}

\begin{proof}
By consistency and the conditional independence assumptions,
\[
E(Y\mid S=1,A=a,\Delta=1,W)=m_a(W).
\]
Within a level $Z=z$, the conditional distribution of $W$ among observed treated outcomes is tilted by
\[
q_a(W)=P(S=1,A=a,\Delta=1\mid W).
\]
Therefore
\[
E(Y\mid S=1,A=a,\Delta=1,Z=z)=\frac{E\{m_a(W)q_a(W)\mid Z=z\}}{E\{q_a(W)\mid Z=z\}}.
\]
Subtracting $E\{m_a(W)\mid Z=z\}=E\{Y(a)\mid Z=z\}$ gives the covariance identity.
\end{proof}

If the target population is restricted to patients already in the EHR frame, the factor $s(W)$ is removed. If all outcomes are observed, $\rho_a(W)$ is removed. In a randomized trial with fixed follow up, only the assignment factor remains, and the identity reduces to the simpler compression drift result for generated covariates.

The theorem gives a precise criterion for information set emulation. A reconstructed information set $Z$ is safe for arm $a$ at $z$ if
\[
\operatorname{Cov}\{m_a(W),s(W)\pi_a(W)\rho_a(W)\mid Z=z\}=0.
\]
Sufficient conditions are
\[
s(W)\pi_a(W)\rho_a(W)=g_a(Z)
\]
or
\[
m_a(W)=h_a(Z).
\]
The first condition preserves the joint selection probability, without requiring each factor to be separately determined by the representation. The second preserves the potential outcome mean and is sufficient on its own. Accurate prediction of the observed outcome alone establishes neither of these exact measurability conditions.

The same identity can be read as a projection statement. Let
\[
m_a^\perp(W)=m_a(W)-E\{m_a(W)\mid Z\},\qquad
q_a^\perp(W)=q_a(W)-E\{q_a(W)\mid Z\}.
\]
Then the numerator of the drift is
\[
\operatorname{Cov}\{m_a(W),q_a(W)\mid Z\}
=E\{m_a^\perp(W)q_a^\perp(W)\mid Z\}.
\]
Thus drift is the conditional inner product between the residual potential outcome mean and the residual joint observation mechanism after projection onto $Z$. Under the theorem's positive-denominator condition, zero arm-specific drift is equivalent to $E\{m_a^\perp(W)q_a^\perp(W)\mid Z\}=0$. Measurability of either $m_a(W)$ or $q_a(W)$ with respect to $Z$ is a sufficient condition, not a necessary one: both residuals can be nonzero and conditionally orthogonal. For example, if $W$ is uniform on $\{-1,0,1\}$, $Z$ is constant, $m_a(W)=W$, and $q_a(W)=0.2+0.2W^2$, both functions are nonconstant but their covariance is zero. Predictive performance alone does not determine this conditional covariance.

\section{Recovery error and drift bounds}

The previous theorem is exact but involves the unobserved prospective information $W$. A useful bound links information recovery to drift. Suppose a validation process, chart review, or prospective sub study gives a reconstruction $\widehat W(Z)$ of $W$ such that
\[
E\{d^2(W,\widehat W(Z))\}\le \eta^2
\]
under the target law. Suppose $m_a$ and $q_a$ are Lipschitz in the same metric:
\[
|m_a(w)-m_a(w')|\le L_{ma}d(w,w'),
\]
\[
|q_a(w)-q_a(w')|\le L_{qa}d(w,w'),
\]
and assume $E\{q_a(W)\mid Z\}\ge c>0$.

\begin{proposition}[Quadratic recovery bound]
Under the conditions above, the integrated absolute compression drift satisfies
\[
E\left[\left|E(Y\mid S=1,A=a,\Delta=1,Z)-E\{Y(a)\mid Z\}\right|\right]
\le
\frac{L_{ma}L_{qa}}{c}\eta^2.
\]
\end{proposition}

\begin{proof}
From the compression drift identity and Cauchy Schwarz,
\[
\left|\frac{\operatorname{Cov}(m_a(W),q_a(W)\mid Z)}{E(q_a(W)\mid Z)}\right|
\le
c^{-1}\{\operatorname{Var}(m_a(W)\mid Z)\operatorname{Var}(q_a(W)\mid Z)\}^{1/2}.
\]
Because $\widehat W(Z)$ is measurable with respect to $Z$,
\[
\operatorname{Var}(m_a(W)\mid Z)\le E\left[\{m_a(W)-m_a(\widehat W(Z))\}^2\mid Z\right]
\le L_{ma}^2 E\{d^2(W,\widehat W(Z))\mid Z\}.
\]
The same bound holds for $q_a$ with constant $L_{qa}$. Taking expectations gives the result.
\end{proof}

The bound gives a sufficient route to small drift when reconstruction controls residual variation in both the potential outcome mean and the joint observation mechanism. Small residuals in both components are not necessary: exact measurability of either function, or conditional covariance cancellation, can also give zero drift.

\section{Post treatment text and source separation}

EHR notes create a distinctive timing problem. A note recorded after treatment may contain pre treatment history and post treatment response in the same document. Excluding all post index notes may throw away useful baseline information. Using the full note may create leakage. The distinction between clinical time and availability time is critical.

This source-separation step is closely related to text distillation in \citet{daoud2022}, which removes treatment-generated textual components while retaining information useful for confounding adjustment. Separating treatment-related text is therefore an established idea. The additional requirement here is to record clinical time and decision-time availability separately and to connect the resulting evidence or unresolved ambiguity to a locked target and its compatible reporting set. Successful distillation alone does not establish that the reconstructed content was available to the original treatment decision.

Let a note be decomposed as
\[
N=N^-\oplus N^+,
\]
where $N^-$ refers to pre treatment clinical content and $N^+$ refers to post treatment course, response, complications, or outcome adjacent documentation. A feature extracted from the whole note is
\[
Z=R_\theta(N^-,N^+).
\]
A source separated feature is
\[
Z^-=R^-_\theta(N^-).
\]

\begin{proposition}[Dual-clock source separation]
Under the primary point-estimation certificate, a feature generated from a post index note is admissible as baseline clinical information only if its output is invariant to the post treatment component relevant to the locked contrast,
\[
R_\theta(N^-,N^+(1))=R_\theta(N^-,N^+(0)),
\]
or if the representation is source separated so that
\[
R_\theta(N^-,N^+)=R^-_\theta(N^-).
\]
It is admissible as treatment decision or design information only under the stronger condition that the certified content was available at or before the decision time:
\[
t^{avail}(Z^-)\le T_i.
\]
If source separation fails, conditioning on the feature compares treatment specific strata. If source separation holds but availability fails, the feature may reconstruct baseline clinical state but not the information set that guided treatment choice.
\end{proposition}

\begin{proof}
If the output depends on $N^+(a)$, the feature has potential values $Z(a)$. By consistency, among those with $A=a$ and observed $Z=z$, the conditioning event is $Z(a)=z$. Among those with $A=1-a$, the conditioning event is $Z(1-a)=z$. Unless the feature is invariant or source separated, these are not a common baseline stratum. The final claim follows from the definition of availability: a fact that existed before treatment but was not available to the decision process cannot be used to reconstruct a decision information set, even if it can be used as baseline clinical history.
\end{proof}

This result is constructive. It does not say that post index notes are unusable. It says that using them for baseline reconstruction requires source separation, and using them for decision reconstruction requires source separation plus availability evidence.

In implementation, source separation should be tested rather than assumed. One architecture neutral procedure is masking. The extractor is applied to the full document, then to the same document with post treatment sentences masked. If the proposed baseline feature changes materially, the certificate fails source separation. A second procedure is span restricted extraction: the model is required to return the source span used for the feature and the span must be classified as pre treatment clinical content by the clock schema in Table~\ref{tab:extraction_schema}. A third procedure is adversarial prompting or review: reviewers ask whether the feature could be inferred from treatment response, discharge status, or endpoint language. These procedures do not prove invariance for all possible notes, but they provide an auditable basis for $B_{\mathrm{sep}}(K)$ and for excluding features that cannot be separated.

\begin{table}[H]
\centering
\caption{Architecture neutral source separation checks for post index text. A feature must pass the relevant checks before it can contribute to the point identified pre treatment information set.}
\label{tab:source_separation_checks}
\scriptsize
\resizebox{\textwidth}{!}{%
\begin{tabular}{>{\raggedright\arraybackslash}p{0.20\textwidth}>{\raggedright\arraybackslash}p{0.33\textwidth}>{\raggedright\arraybackslash}p{0.22\textwidth}>{\raggedright\arraybackslash}p{0.20\textwidth}}
\toprule
Check & Procedure & Passing rule & Failure consequence \\
\midrule
Span restriction & Require the extractor to return the exact sentence or event span used for the feature & The span is classified as pre treatment clinical content under the clock schema & Route to mixed role or exclude from point estimator \\
Post treatment masking & Re run the extractor after masking post treatment sentences, treatment response, discharge state, and endpoint language & Feature value and proposed role remain stable within a prespecified tolerance & Add to $B_{\mathrm{sep}}$ or move to $J_{amb}(K)$ \\
Availability adjudication & Ask whether the source supports availability at the treatment decision time, not merely pre treatment existence & $t^{avail}\le T_i$ is supported by source, protocol, or workflow evidence & May be baseline clinical state but not design or decision information \\
Outcome adjacency check & Screen source spans for endpoint definitions, improvement language, mortality, discharge status, or follow up summary & Outcome relation is independent or unresolved but bounded & Treat as outcome proxy risk and exclude from point estimator \\
Repeated extraction and role review & Repeat extraction across seeds, folds, prompts, or reviewers & Role disagreement and feature instability are below prespecified thresholds & Expand compatible role set or sensitivity envelope \\
\bottomrule
\end{tabular}%
}
\end{table}

\subsection*{Synthetic certificate vignette}

The present paper does not use real patient data. To show the intended workflow, Table~\ref{tab:synthetic_vignette} gives a synthetic note example. The point is not the clinical content; it is the certificate logic.

\begin{table}[H]
\centering
\caption{Synthetic note vignette illustrating certificate based routing. The text is artificial and contains no patient data.}
\label{tab:synthetic_vignette}
\scriptsize
\resizebox{\textwidth}{!}{%
\begin{tabular}{>{\raggedright\arraybackslash}p{0.26\textwidth}>{\raggedright\arraybackslash}p{0.18\textwidth}>{\raggedright\arraybackslash}p{0.18\textwidth}>{\raggedright\arraybackslash}p{0.16\textwidth}>{\raggedright\arraybackslash}p{0.16\textwidth}}
\toprule
Synthetic source text & Candidate feature & Certificate finding & Primary role & Consequence \\
\midrule
Pre index clinic note: patient walks one block, needs help with bathing, ECOG not recorded & AI performance status & Clinical and recording times before $T_i$; available before decision & $B$ or $D$ depending on target lock & May enter certified information set \\
Post index progress note: patient looked frail before treatment and improved after therapy & AI frailty summary & Mixed $N^-$ and $N^+$; source separation required & unresolved $B/M/Y$ & Exclude from point identified primary estimator unless separated \\
Discharge summary: poor prognosis discussed, died during admission & AI prognosis score & Outcome adjacent source; endpoint information present & $Y$ & Exclude from baseline adjustment \\
Order history: extra labs ordered every 6 hours after treatment & Monitoring intensity & Observation mechanism after treatment & $\mathsf{Obs}$ & Use in observation model or sensitivity analysis, not as baseline state \\
\bottomrule
\end{tabular}%
}
\end{table}
\section{Certificate conditional approximation frontier}

The value of information set emulation should not be described only as credibility. It should be evaluated by error. Fix a declared world class $\Theta$, typed evidence $K$, and an observed law $p$ with nonempty compact compatible image $\mathfrak P_K(p)$. Let $r_K(p)$ be its information radius and let $c_K(p)$ be a Chebyshev center.

Let $\widehat c_K$ estimate $c_K(p)$ from $n_E$ observations generated under the common observed law $p$, and suppose
\[
b_n=E_p\{\widehat c_K-c_K(p)\},\qquad \operatorname{Var}_p(\widehat c_K)=\frac{\sigma_K^2}{n_E}+o(n_E^{-1}).
\]

\begin{theorem}[Certificate conditional approximation frontier]
Under the fixed tuple $(\Theta,K,p)$ above,
\[
\sup_{\theta\in\mathcal F_K(p)}E_\theta\{(\widehat c_K-\Psi(\theta))^2\}
\le
r_K^2(p)+2r_K(p)|b_n|+b_n^2+\frac{\sigma_K^2}{n_E}+o(n_E^{-1}).
\]
Consequently, by the triangle inequality in $L_2$,
\[
\sup_{\theta\in\mathcal F_K(p)}\{E_\theta(\widehat c_K-\Psi(\theta))^2\}^{1/2}
\le
r_K(p)+\left(b_n^2+\frac{\sigma_K^2}{n_E}\right)^{1/2}+o(n_E^{-1/2}).
\]
Let a prospective study estimator have mean squared error approximately $\sigma_P^2/n_P$. Define the nonnegative tolerance $\tau_{\mathrm{MSE}}$ on the mean-squared-error scale. The AI augmented EHR study satisfies the prospective-study approximation criterion at that tolerance if
\[
\left\{r_K(p)+\left(b_n^2+\frac{\sigma_K^2}{n_E}\right)^{1/2}\right\}^2
\le
\frac{\sigma_P^2}{n_P}+\tau_{\mathrm{MSE}}.
\]
\end{theorem}

\begin{proof}
For any compatible $\psi$,
\[
\widehat c_K-\psi=\{\widehat c_K-c_K(p)\}+\{c_K(p)-\psi\}.
\]
Let $d_\psi=c_K(p)-\psi$, so $|d_\psi|\le r_K(p)$. Then
\[
E_p\{(\widehat c_K-\psi)^2\}
=E_p\{(\widehat c_K-c_K(p))^2\}+2d_\psi E_p\{\widehat c_K-c_K(p)\}+d_\psi^2.
\]
Because every $\theta\in\mathcal F_K(p)$ induces the same observed-data law $p$, the distribution of $\widehat c_K$ is common across the fiber. Using $E_p\{(\widehat c_K-c_K(p))^2\}=b_n^2+\sigma_K^2/n_E+o(n_E^{-1})$ and taking the supremum over the compatible image gives the first bound. The root mean squared error inequality is a weaker consequence obtained by the triangle inequality in $L_2$; it is not equivalent to the first display. Squaring that RMSE bound and comparing it with the prospective mean squared error gives the stated criterion on the MSE scale.
\end{proof}

This theorem is not a claim that retrospective EHR studies equal prospective studies. It states a decision criterion. A large EHR study can approach the mean squared error of a smaller prospective study only if the information radius and estimation bias are small enough. When the certificate leaves substantial role ambiguity or unmeasured decision information, $r(K)$ remains large and the study is not a prospective study approximation.

\section{Simulation study}

The simulation is artificial by design. It is not intended to show the effect of a clinical treatment and is not external validation of EHR feasibility. It is a controlled failure mode experiment for the mathematical claims above in a setting where the prospective information set, causal type, potential outcomes, EHR presence, treatment choice, and outcome observation mechanisms are all known. Each method produces a causal effect estimate, an estimated standard error, and a confidence interval. The comparison is therefore not whether AI can prepare better features, but whether the information set created by each method supports the intended causal effect estimator. The main simulations use standardized working models to isolate information-structure failures. They are not intended as a numerical study of the semiparametric efficiency result in Theorem~\ref{thm:orthogonal}. The AIPW route is one valid estimation path once a point information set has been certified; the simulations ask a different question, namely whether the information structure supplied to the estimator preserves the locked target.

The design follows the ADEMP logic for simulation studies: aims, data generating mechanisms, estimands, methods, and performance measures are specified before results are read. We used 1000 Monte Carlo replications with $n=1000$ generated patients per replication. The Monte Carlo standard error of a nominal 0.95 coverage estimate is approximately $\{0.95(0.05)/1000\}^{1/2}=0.0069$.

\subsection{Data generating mechanism}

For each patient we generated a prospective information set
\[
W_i=(X_i,D_i,J_i,O_i).
\]
Here $X_i$ is clinical state, $D_i$ is a design, eligibility, or treatment availability indicator, $J_i$ is latent clinical judgement such as performance status, and $O_i$ is visit or measurement intensity. The structured EHR variables recorded $X_i$, $D_i$, and $O_i$, but not $J_i$. The AI reconstruction was generated as a noisy pre treatment proxy for $J_i$. A separate post index feature encoded treatment response and outcome adjacent documentation.

Presence in the EHR analytic frame, treatment choice, and outcome observation were generated by
\[
\begin{aligned}
s(W_i)&=P(S_i=1\mid W_i),\\
\pi_1(W_i)&=P(A_i=1\mid S_i=1,W_i),\\
\rho_a(W_i)&=P(\Delta_i=1\mid S_i=1,A_i=a,W_i).
\end{aligned}
\]
The potential outcome model was
\[
Y_i(a)=0.30+0.75X_i+0.80D_i+1.00J_i+0.65O_i+a(0.60+0.30J_i)+\varepsilon_i.
\]
Thus the empirical prospective target among patients in the EHR frame is
\[
\Psi_P=n_S^{-1}\sum_{i:S_i=1}\{0.60+0.30J_i\}.
\]
The fitted models used outcome data only when $S_i=\Delta_i=1$ but standardized predictions over the EHR frame $S_i=1$.

In this main design the AI reconstruction is generated with low extraction noise and correctly certified roles. It is intentionally a favorable setting; the stress tests in Section~11.5 relax exactly these conditions. We compared six causal estimators. The Structured EHR estimator adjusted only for recorded structured variables and therefore omitted the latent clinical judgement $J_i$. The Naive AI estimator placed all AI features into the adjustment set, including post index and design erasing signals. The ISE causal estimator used the certified pre treatment AI reconstruction while retaining design and observation relevant structured information. The Design erasure estimator used an AI score that compressed clinical judgement while dropping design and observation information. The Post index leakage estimator used a post treatment note feature as if it were baseline information. The Oracle prospective estimator used the true $W_i$ and is included only as a benchmark for the prospective information set.

\begin{table}[H]
\centering
\caption{Main artificial EHR simulation. Results use 1000 Monte Carlo replications with $n=1000$ generated patients per replication. The target is the empirical prospective information set effect among patients in the EHR frame.}
\label{tab:main_sim}
\resizebox{\textwidth}{!}{%
\begin{tabular}{lrrrrrr}
\toprule
Method & True & Estimate & Bias & RMSE & Mean SE & Coverage \\
\midrule
Structured EHR estimator & 0.644 & 1.391 & 0.747 & 0.761 & 0.140 & 0.001 \\
Naive AI estimator & 0.644 & 0.140 & -0.503 & 0.513 & 0.095 & 0.000 \\
ISE causal estimator & 0.644 & 0.657 & 0.013 & 0.115 & 0.111 & 0.938 \\
Design erasure estimator & 0.644 & 1.026 & 0.383 & 0.404 & 0.126 & 0.146 \\
Post index leakage estimator & 0.644 & 0.164 & -0.479 & 0.491 & 0.100 & 0.000 \\
Oracle prospective estimator & 0.644 & 0.645 & 0.002 & 0.114 & 0.110 & 0.941 \\
\bottomrule
\end{tabular}%
}
\end{table}

The main pattern is the intended one. Structured EHR estimation retained large residual bias because the clinical judgement variable was not observed. The ISE causal estimator recovered nearly the oracle performance, with small bias and approximately nominal coverage. Naive AI estimation, design erasure, and post index leakage did not merely add noise; they changed the causal information set and moved the estimator toward the wrong target. The post index leakage estimator was precise but wrong, which is the practical danger of AI features derived from outcome adjacent text.

\subsection{Flat table nonidentification and certificate shrinkage}

Corollary~\ref{cor:flat_nonid} shows that the causal target need not be identified from a flat analysis table when its compatible worlds imply different target values. Table~\ref{tab:flat_nonid} gives a binary artificial example. The same observed distribution of $(A,Z,Y)$ is compatible with $Z$ as a baseline state, a post treatment mediator, or an outcome-adjacent proxy, but the corresponding estimands differ.

\begin{table}[H]
\centering
\caption{Flat table nonidentification demonstration. The same observed distribution of $(A,Z,Y)$ admits incompatible causal roles for $Z$ and incompatible estimands.}
\label{tab:flat_nonid}
\resizebox{\textwidth}{!}{%
\begin{tabular}{llr}
\toprule
Interpretation & Estimand & Value \\
\midrule
Z as baseline state & Standardized contrast over $P(Z)$ & 0.250 \\
Z as post treatment mediator & Crude total contrast under exogenous treatment & 0.460 \\
Z as outcome adjacent proxy & Sensitivity-completed treatment contrast & -0.170 \\
Difference from same flat table & Mediator interpretation minus baseline interpretation & 0.210 \\
\bottomrule
\end{tabular}%
}
\end{table}

Certification narrows the compatible estimand set by excluding causal worlds that conflict with the source and timing evidence. Table~\ref{tab:certificate_shrinkage} shows the induced reduction in the information radius for the same flat table demonstration.

\begin{table}[H]
\centering
\caption{Certificate induced shrinkage of compatible estimands in the flat table demonstration. The squared radius is the exact minimax information risk for any estimator that cannot distinguish the remaining compatible worlds.}
\label{tab:certificate_shrinkage}
\resizebox{\textwidth}{!}{%
\begin{tabular}{lrrrrr}
\toprule
Certificate & Number & Lower & Upper & Radius & Information risk $r^2$ \\
\midrule
No temporal certificate $K_0$ & 3 & -0.170 & 0.460 & 0.315 & 0.0992 \\
AI typed lift $K_\theta$ & 2 & 0.250 & 0.460 & 0.105 & 0.0110 \\
Pre treatment availability certificate $K_\theta\vee E$ & 1 & 0.250 & 0.250 & 0.000 & 0.0000 \\
Post index source certificate $K_{post}$ & 2 & -0.170 & 0.460 & 0.315 & 0.0992 \\
\bottomrule
\end{tabular}
}
\end{table}

The three compatible worlds in the untyped row are the baseline, mediator, and outcome-adjacent constructions of Table~\ref{tab:flat_nonid}. The typed lift excludes the outcome-adjacent world and shrinks the radius from $0.315$ to $0.105$; the source and availability audit isolates the pre-treatment world and reduces the radius to zero. The post-index certificate excludes the baseline world but retains the two extremal targets, so its radius does not shrink. This equality case is important: stronger evidence can remove worlds without reducing the Chebyshev radius when the removed worlds are not extremal. The construction and all reported values are generated by the exact compatible-world script used again in Section~\ref{sec:finite_frontier}.

\subsection{Compression drift identity}

We verify the identity by exact enumeration on a finite artificial population. Independently, $X,J,O$ are uniform on $\{-1,1\}$ and $D$ is Bernoulli with probability $1/2$, giving 16 equally likely states. The functions $s(W)$, $\pi_a(W)$, $\rho_a(W)$, and $m_a(W)$ use the coefficients in Appendix~A. This finite-support validation is separate from the continuous simulation in Table~\ref{tab:main_sim}. The compressed feature is $Z=X+J$; the richer reconstruction is $(Z,D,O)$.

For each stratum and arm, the direct drift is
\[
d_a(z)=
\frac{\sum_w p(w\mid z)q_a(w)m_a(w)}{\sum_w p(w\mid z)q_a(w)}
-\sum_w p(w\mid z)m_a(w).
\]
Table~\ref{tab:compression_drift} reports the population-stratum-probability weighted absolute drift and the same summary computed from the conditional covariance formula. The final column is the maximum absolute difference over all strata. The compressed and richer reconstructions have 3 and 12 strata, respectively; none is omitted. All conditional quantities are enumerated, so differences reflect only floating-point arithmetic. The script \texttt{compressionDriftValidation.py} writes the support, all cell calculations, summary, and table.

\begin{table}[H]
\centering
\caption{Exact finite-support validation of the EHR compression-drift identity. Absolute drifts are averaged using population stratum probabilities. There is no Monte Carlo or binning error.}
\label{tab:compression_drift}
\resizebox{\textwidth}{!}{%
\begin{tabular}{lrrrr}
\toprule
Feature & Arm & Direct abs drift & Formula abs drift & Max error \\
\midrule
Compressed score & 0 & 0.202502 & 0.202502 & $<10^{-12}$ \\
Compressed score & 1 & 0.421241 & 0.421241 & $<10^{-12}$ \\
Richer reconstruction strata & 0 & 0.036291 & 0.036291 & $<10^{-12}$ \\
Richer reconstruction strata & 1 & 0.040359 & 0.040359 & $<10^{-12}$ \\
\bottomrule
\end{tabular}
}
\end{table}

\subsection{Prospective approximation frontier}\label{sec:finite_frontier}

Finally, we calibrated the certificate conditional approximation frontier on an explicit finite fiber. This calibration is separate from the estimator-performance comparison in Table~\ref{tab:main_sim}. It fixes one strictly positive observed law $p(A,Z,Y)$:
\[
P(Z=1)=0.5,\qquad P(A=1\mid Z=0)=0.2,\qquad P(A=1\mid Z=1)=0.8,
\]
and
\[
\begin{array}{c|cc}
 & Z=0 & Z=1\\ \hline
P(Y=1\mid A=0,Z) & 0.10 & 0.40\\
P(Y=1\mid A=1,Z) & 0.30 & 0.70
\end{array}.
\]
Three causal worlds induce exactly this same observed law. In world $B$, $Z$ is a pre-treatment state and standardization over $P(Z)$ gives $\Psi_B=0.250$. In world $M$, treatment is exogenous and $Z$ is post-treatment, so the total contrast is the crude contrast $\Psi_M=0.460$. In world $Y$, $Z$ is outcome-adjacent and treatment selection is latent; the observed potential-outcome means are completed by the declared sensitivity coordinates $E\{Y(1)\mid A=0\}=0.10$ and $E\{Y(0)\mid A=1\}=0.90$, giving $\Psi_Y=-0.170$. The latter construction satisfies consistency and reproduces the same observed $(A,Y)$ law; drawing $Z$ from the common $P(Z\mid A,Y)$ reproduces the full observed law. Thus these are not estimator biases used as substitutes for a radius: they are estimand values of explicitly compatible worlds in the same fiber.

The flat-table certificate $K_0$ retains $\{B,M,Y\}$; the typed lift $K_\theta$ excludes the outcome-adjacent world but retains $\{B,M\}$; the source and availability audit $K_\theta\vee E$ certifies $B$; and a post-index source certificate $K_{post}$ retains $\{M,Y\}$. The exact compatible image, Chebyshev center, and radius are computed directly for each state. In 5000 Monte Carlo samples of size $n_E=1000$, plug-in estimates of the compatible-world functionals were used only to estimate the sampling RMSE of the Chebyshev center. The observed worst-world RMSE is $\sup_{\theta\in\mathcal F_K(p)}\{E_\theta(\widehat c_K-\Psi(\theta))^2\}^{1/2}$. The prospective reference estimates $\Psi_B$ from independent samples of size $n_P=300$. Table~\ref{tab:frontier} reports all quantities on the RMSE scale; the tolerance is zero.

\begin{table}[H]
\centering
\caption{Exact compatible-world certificate frontier. The radius is the Chebyshev radius of an explicitly enumerated compatible estimand image under a common observed law. Center RMSE is Monte Carlo sampling error; worst-world RMSE is evaluated over the retained fiber. The prospective reference RMSE is $0.068$ and the tolerance is zero.}
\label{tab:frontier}
\scriptsize
\resizebox{\textwidth}{!}{%
\begin{tabular}{llrrrrr}
\toprule
Evidence & Retained worlds & $r_K(p)$ & Center RMSE & RMSE bound & Worst-world RMSE & Meets \\
\midrule
Flat $K_0$ & B, M, Y & 0.315 & 0.020 & 0.335 & 0.316 & No \\
Typed $K_\theta$ & B, M & 0.105 & 0.030 & 0.135 & 0.110 & No \\
Typed + audit $K_\theta\!\vee\!E$ & B & 0.000 & 0.036 & 0.036 & 0.036 & Yes \\
Post-index $K_{post}$ & M, Y & 0.315 & 0.020 & 0.335 & 0.316 & No \\
\bottomrule
\end{tabular}%
}
\end{table}

\begin{figure}[H]
\centering
\includegraphics[width=0.95\textwidth]{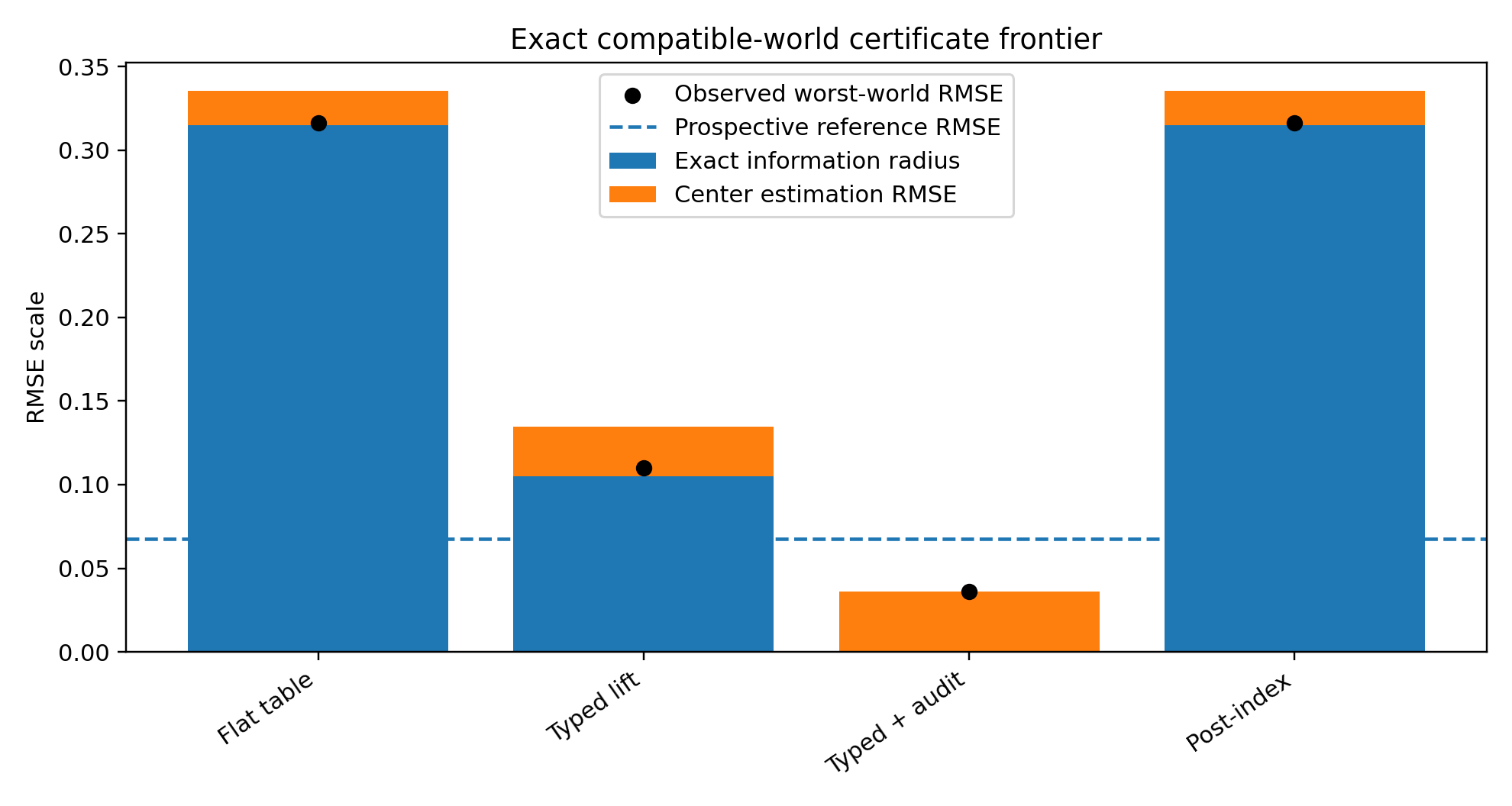}
\caption{Exact compatible-world certificate frontier. Bars stack the exact information radius and the Monte Carlo RMSE for estimating its Chebyshev center. Points show the observed worst-world RMSE; the dashed line is the prospective reference RMSE. The source and availability audit meets the zero-tolerance criterion because it reduces the compatible image to the certified pre-treatment world.}
\label{fig:frontier}
\end{figure}

The exact calibration separates information ambiguity from center-estimation error. The flat table and post-index certificate remain outside the frontier because incompatible causal worlds survive even with large $n_E$. The typed lift narrows the fiber but does not meet the zero-tolerance reference until the source and availability audit resolves the remaining baseline-versus-mediator ambiguity. No quantity in Table~\ref{tab:frontier} uses absolute estimator bias as a proxy for $r_K(p)$.

\subsection{Certificate and extraction stress tests}

The main simulation gives the method a favorable setting: the pre treatment AI reconstruction is accurate and correctly certified. To test the boundaries, we added artificial stress tests with 200 Monte Carlo replications per setting. These scenarios vary extraction error, partial post index leakage, role misclassification, availability error, and positivity stress. The settings are intentionally stylized controlled failure mode experiments, not external validation of EHR feasibility. Their purpose is to show when the causal estimator stops behaving like the oracle prospective estimator and when the typed lift must route the feature to a compatible set instead of a point estimator.

\begin{table}[H]
\centering
\caption{Additional artificial stress tests for certificate and extraction failures. Results use 200 Monte Carlo replications with $n=1000$ generated patients per replication.}
\label{tab:additional_stress}
\resizebox{\textwidth}{!}{%
\begin{tabular}{llrrr}
\toprule
Failure mode & Setting & Bias & RMSE & Coverage \\
\midrule
AI extraction error & low extraction error & 0.005 & 0.103 & 0.980 \\
AI extraction error & moderate extraction error & 0.266 & 0.293 & 0.435 \\
AI extraction error & high extraction error & 0.497 & 0.514 & 0.050 \\
Post index partial leakage & no leakage lambda 0 & 0.011 & 0.114 & 0.935 \\
Post index partial leakage & partial leakage lambda 0.25 & -0.308 & 0.323 & 0.115 \\
Post index partial leakage & partial leakage lambda 0.50 & -0.423 & 0.433 & 0.005 \\
Certificate role misclassification & misclassification 0 percent & 0.024 & 0.112 & 0.970 \\
Certificate role misclassification & misclassification 10 percent & -0.169 & 0.202 & 0.625 \\
Certificate role misclassification & misclassification 20 percent & -0.305 & 0.319 & 0.105 \\
Availability error & available at time zero & 0.017 & 0.106 & 0.950 \\
Availability error & late availability lambda 0.50 & -0.818 & 0.822 & 0.000 \\
Availability error & late availability lambda 1.00 & -1.103 & 1.106 & 0.000 \\
Positivity stress & ordinary overlap & 0.011 & 0.116 & 0.920 \\
Positivity stress & strong positivity stress & 0.023 & 0.140 & 0.940 \\
\bottomrule
\end{tabular}%
}
\end{table}

The stress tests emphasize that information set emulation is conditional. Moderate and high extraction error widened bias and reduced coverage. Partial leakage moved the estimator away from the target even when the feature still contained pre treatment signal. Availability error was especially damaging: a fact that exists before treatment but is only learned through late or outcome adjacent documentation can behave like a post index feature when used to reconstruct the treatment decision set. Positivity stress increased variability even when the role was certified. Role misclassification was less damaging in this particular design because the downstream domination rule kept some ambiguous features out of the point information set and because the misclassified features carried less treatment decision information than the availability variables. The result should not be read as evidence that role errors are harmless; it shows that different certificate failures act through different parts of the information structure. These results support the stopping rule implied by the theory: if the certificate cannot bound extraction, availability, leakage, and overlap errors, the analysis should report a compatible set or remain exploratory.

\section{Phase~0 validation of the typed lift}

The paper is a theory and artificial simulation study, but the proposed method is meant to be used. Before a full causal effect analysis in a hospital system, the AI typed lift should undergo a Phase~0 validation that does not estimate the treatment effect. Its purpose is to ask whether the typed ledger can be produced reliably enough to define $\mathcal H_{\theta,K}^{point}$ and $\mathfrak P(K)$.

A minimal Phase~0 study samples notes or event streams from the intended target trial window, applies the typed lift, and has clinical or epidemiologic reviewers adjudicate the evidence. The recommended outputs are source pointer precision, clock extraction agreement, availability adjudication agreement, role agreement, adjudication disagreement, false admissible rate, source separation failure rate, extractor instability, and, when a common-law world model and validated exclusion rules are available, the corresponding information-radius change. Table~\ref{tab:phase0} gives the minimal reporting quantities.

\begin{table}[H]
\centering
\caption{Minimal Phase~0 validation outputs for an AI typed EHR lift. The goal is feasibility of information set emulation, not estimation of a clinical treatment effect.}
\label{tab:phase0}
\scriptsize
\resizebox{\textwidth}{!}{%
\begin{tabular}{>{\raggedright\arraybackslash}p{0.23\textwidth}>{\raggedright\arraybackslash}p{0.42\textwidth}>{\raggedright\arraybackslash}p{0.27\textwidth}}
\toprule
Quantity & Definition & Use in the framework \\
\midrule
Source pointer precision & Fraction of AI cited source spans that reviewers judge to support the extracted feature & Informs $B_{\mathrm{extract}}$ and source audit quality \\
Clock extraction agreement & Agreement for clinical time, measurement time, recording time, and availability time & Informs temporal and availability predicates \\
Availability agreement & Agreement that information was available at the treatment decision time & Determines whether a feature can be decision or design information \\
Role agreement & Agreement on baseline, design, history, marker, mediator, outcome proxy, observation, intercurrent, or excluded role & Determines role set size and false admissible risk \\
Adjudication disagreement & Fraction of features requiring adjudication or retaining unresolved reviewer disagreement after adjudication; kappa or similar agreement summaries may be reported as descriptive diagnostics & Contributes to $B_{\mathrm{judge}}$ or routes features to $J_{\mathrm{amb}}(K)$ rather than the point estimator \\
False admissible rate & Fraction of features proposed for $\mathcal A_{pre}$ that reviewers classify outside $\mathcal A_{pre}$ & If large, point estimation should be replaced by compatible reporting \\
Source separation failure rate & Fraction of late-note features that change after post treatment masking or fail span restriction & Informs $B_{\mathrm{sep}}$ and exclusion rules \\
Extractor instability & Feature or role variation across seeds, folds, prompts, or model versions & Informs $B_{\mathrm{extract}}$, value uncertainty, and role uncertainty \\
Radius reduction & Reduction in $\widehat r_{report}(K)$ after AI typing, review, or adjudication & Estimates value of typing and certification \\
\bottomrule
\end{tabular}%
}
\end{table}

This validation is intentionally smaller than an applied causal study. It can be performed on tens to hundreds of records, depending on the complexity of the features. A successful Phase~0 study does not prove exchangeability or treatment effectiveness. It shows that the raw EHR can be lifted into a typed information ledger with quantified residual ambiguity. Inter-reviewer agreement statistics, such as kappa-type summaries, are useful quality diagnostics, but they are not causal certificates by themselves. Persistent disagreement is treated as unresolved role or availability uncertainty and is carried into $B_{\mathrm{judge}}$, $J_{\mathrm{amb}}(K)$, or the compatible reporting set. For medical informatics and real world evidence applications, Phase~0 is the minimum empirical feasibility test of the typed lift. If Phase~0 fails, the appropriate conclusion is not that AI is useless; it is that the retrospective study is not yet close enough to the prospective information set for the intended causal claim.

\subsection{Synthetic Phase~0 note validation}

We added a reproducible synthetic Phase~0 validation to illustrate the workflow in Table~\ref{tab:phase0}. This is not a deidentified EHR study and not a treatment-effect analysis. The note-like records are generated from the artificial data generating process in the supplementary script. Each synthetic patient has a pre-index clinic note, order history, discharge-summary history statement, post-index progress note, and monitoring-intensity record. The generator assigns the true source, clinical time, recording time, availability status, and causal role for each candidate feature. The typed lift attempts to recover those fields with prespecified error probabilities, and the audit stage simulates improved field recovery. Point admission in this diagnostic uses the proposed role set alone; it does not implement the complete certificate predicates, availability gate, or outcome-analysis workflow. The false-admissible proportion is conditional on that role-only admission.

The validation measures how accurately the event stream is lifted into a typed ledger. A separate, prespecified routing illustration summarizes the spread of role-specific analyses. It does not use patient data and does not claim external EHR feasibility. Its value is reproducibility: the source text, extraction errors, adjudication errors, and role-specific analysis values are all generated by declared mechanisms. The retained analysis lists below are stipulated; they are not inferred from the measured audit error rates.

\begin{table}[H]
\centering
\caption{Synthetic Phase~0 note validation. The artificial validation contains 120 generated patients and 600 candidate feature records. Values are proportions. The audit stage represents source-span review, clock adjudication, and role review under the declared synthetic error model.}
\label{tab:synthetic_phase0_metrics}
\footnotesize
\begin{tabular}{lccccccc}
\toprule
Stage & Source & Clock & Avail. & Role cov. & False adm. & Sep. fail & Instab. \\
\midrule
AI typed lift & 0.953 & 0.888 & 0.878 & 0.958 & 0.042 & 0.192 & 0.188 \\
AI plus Phase~0 audit & 0.988 & 0.972 & 0.973 & 0.977 & 0.016 & 0.033 & 0.000 \\
\bottomrule
\end{tabular}
\end{table}

Table~\ref{tab:synthetic_phase0_metrics} shows the intended use of Phase~0. The AI typed lift already recovers most source and role evidence, but residual ambiguity remains. Audit mainly improves clock, availability, and source-separation evidence, which are the fields most likely to determine whether a feature belongs in the point information set or in the compatible set.

For a separate routing illustration, the same artificial design supplies four analysis values: a total-effect target of $0.842$, a direct-effect target of $0.550$, an outcome-proxy-adjusted regression coefficient of $0.102$, and an observation-weighted effect of $0.928$. These quantities answer different questions within one data-generating system. They are not values of one locked estimand across observationally indistinguishable worlds. We stipulate three nested lists of analyses to illustrate how routing changes their numerical spread; the lists are not calibrated from Table~\ref{tab:synthetic_phase0_metrics}.

\begin{table}[H]
\centering
\caption{Prespecified role-specific analysis spread accompanying synthetic Phase~0 validation. The initial list contains all four analyses; the second omits outcome-proxy adjustment; the third retains the total and observation-weighted effects. Half-ranges describe these lists and are not exact fiber radii or minimax risks.}
\label{tab:synthetic_phase0_spread}
\scriptsize
\begin{tabular}{lccccc}
\toprule
Illustrative routing stage & Lower & Upper & Center & Half-range & Squared half-range \\
\midrule
Flat table $K_0$ & 0.102 & 0.928 & 0.515 & 0.413 & 0.170 \\
AI typed lift $K_\theta$ & 0.550 & 0.928 & 0.739 & 0.189 & 0.036 \\
AI plus Phase~0 audit $K_\theta\vee E$ & 0.842 & 0.928 & 0.885 & 0.043 & 0.002 \\
\bottomrule
\end{tabular}
\end{table}

Table~\ref{tab:synthetic_phase0_spread} is a descriptive routing example. Its decreasing half-range does not establish that the measured audit performance shrinks a common-law observational fiber or reduces exact minimax risk. Such a claim would require a single locked estimand, explicitly compatible worlds, and a validated link from audit evidence to world exclusion. The separate construction in Section~\ref{sec:finite_frontier} provides the common-law finite-world illustration; the synthetic notes here demonstrate audit diagnostics only.

\section{Discussion}

Information set emulation treats the EHR record as an incomplete information structure rather than as a covariate table. The typed lift refines that structure by adding source, clock, availability, version, and role evidence. The induced observational fiber determines which prospective causal worlds remain possible; its estimand image determines whether the analysis can support a point claim or must report a set. Conditional on the typed evidence available to the analyst, the squared radius of that image is the exact minimax information-structure risk for the locked scalar target.

The construction is not limited to EHRs. The same mathematics applies to retrospective event streams, including claims data, registries, wearable streams, digital platform logs, and administrative records, whenever the scientific goal is to emulate a prospective information set rather than merely enrich a covariate table. EHRs are especially demanding because clinical facts, recording times, availability, observation processes, and outcome-adjacent narratives are routinely entangled.

The certificate is the operational link between the event stream and the fiber. Weak certificates leave many causal worlds compatible with the same EHR table. Stronger certificates narrow the fiber by documenting clocks, sources, availability, role evidence, version stability, and extraction stability. A wide radius is therefore a design diagnosis: it identifies whether the next investment should be chart review, availability adjudication, source separation, observation-process modeling, external validation, repeated extraction, or a prospective substudy.

Uncertainty is reported differently under this framework. A feature with unresolved causal type does not yield a single target by default. It yields a compatible estimand set. The experiment comparison in Section~4 gives the broad design interpretation, while the information radius gives the estimand-specific diagnostic used for reporting. AI typed lifting should be evaluated by how it changes that set, not only by how well its features predict outcomes. A narrow radius matters only when the supporting evidence has not suppressed plausible downstream roles; a wide radius matters because it shows which information is missing.

\subsection*{Relation to adjacent methods}

The closest literatures address overlapping but distinct layers of the problem. Target trial emulation fixes treatment strategies, time zero, eligibility, outcomes, follow up, and analysis \citep{hernanrobins2020,rctduplicate2021,wang2023,target2025}. Importantly, the distinction between conceptual target-trial specification and its operational realization in EHR data is not new to this paper. \citet{wang2026} formulate EHR target trial emulation as a data-constrained operational design problem and explain when the available EHR cannot support the intended estimand. \citet{rafalko2025} describe how natural language processing can make target-trial components available from unstructured EHR data. Information set emulation builds on that premise rather than claiming it as its novelty.

Text-based causal inference already uses clinical narratives to address missing data, confounding, and treatment-effect heterogeneity \citep{mozer2024}, while representation methods learn or adapt text features for causal adjustment or prediction \citep{veitch2020,zeng2022}. Recent work directly evaluates strategies for inserting LLM-extracted clinical covariates into propensity-score and doubly robust pipelines \citep{liu2026}. Methods for generated regressors, AI-generated covariates, machine-learning predictions as covariates, and prediction-powered inference study uncertainty after a generated variable has been accepted for an analytic role \citep{battaglia2024,fongtyler2021,angelopoulos2023}.

The narrower contribution here is to represent source, clinical time, recording time, decision-time availability, observation status, and causal role in one auditable certificate; to let unresolved role evidence define a joint observational fiber and compatible estimand image; and to place validation and stopping rules on the estimand scale. The information radius is not a new standard-error correction and the Chebyshev identity is not claimed as a new general theorem of partial identification. The contribution is the integrated EHR information structure in which those classical tools become operational. This boundary also distinguishes the present paper from work that evaluates whether an accepted LLM covariate improves an estimator: the certificate asks first whether the feature is admissible for the locked causal role at all.

\subsection*{Auditing published studies and reverse trial reports}

Information set emulation can also be used as a retrospective audit of published observational studies. The goal is not to assign a probability that a published conclusion is correct. It is to reconstruct the documented information structure in the paper and to ask which causal roles, source windows, availability claims, observation processes, and design variables remain unresolved. Let $K_{\mathrm{paper}}$ denote the typed evidence that can be recovered from a published protocol, covariate definitions, source descriptions, timing statements, analysis plan, and reported summaries. Let $\mathcal S_{\mathrm{paper}}$ be the set of observed-data laws compatible with the reported tables, estimates, and sampling description. A documented paper fiber can be written as
\[
\mathcal F_{\mathrm{paper}}=
\{\theta\in\Theta: \Pi_{K_{\mathrm{paper}}}(\theta)\in\mathcal S_{\mathrm{paper}}\}.
\]
For a locked claim, the documented compatible estimand image is
\[
\mathfrak P_{\mathrm{paper}}=\Psi\{\mathcal F_{\mathrm{paper}}\},
\qquad
r_{\mathrm{paper}}=\inf_{c\in\mathbb R}
\sup_{\psi\in\mathfrak P_{\mathrm{paper}}}|\psi-c|.
\]
This paper audit radius is not a score for whether the published conclusion is true. It is a measure of how much target ambiguity remains in the reported information itself. A large radius means that the publication does not distinguish among causally different interpretations, for example baseline adjustment, mediator adjustment, outcome-proxy adjustment, or observation-process modeling. It does not prove the original claim false; it identifies which additional source, timing, availability, or role evidence would be needed before the reported point claim could be read as a locked causal estimand.

Paper audit radii are therefore documented-information radii. They are conditional on $K_{\mathrm{paper}}$, $\mathcal S_{\mathrm{paper}}$, and the declared identification and sensitivity model $\Theta$. They do not recover source documents, time stamps, adjudication records, or design information that may have been available to the original investigators but was not reported. When only summary tables and model descriptions are available, the operational audit object is a conservative reporting set obtained by varying the role assignments and validation-calibrated sensitivity envelopes that remain compatible with the publication. In that setting $r_{\mathrm{paper}}$ is best read as a model-conditional audit quantity: it makes the assumptions and unresolved information visible on the estimand scale, but it is not an automatic objective distance from truth. The present paper defines this audit object but does not perform a full audit of a published observational study. Such an application requires a separate extraction of the published protocol, covariate definitions, timing statements, source descriptions, and reported summaries, and should be reported as an applied audit rather than as evidence that the original conclusion is true or false.

Published randomized-trial reports and protocols play a different role. They do not determine the EHR radius. They help define the prospective information reference $K_P$. A reverse trial audit parses the trial report into eligibility criteria, baseline measurements, stratification variables, treatment milestones, observation schedules, intercurrent-event strategies, missingness rules, subgroup variables, and outcome definitions. These elements become the checklist against which the retrospective EHR typed lift is audited. Trial reports define the prospective information reference; EHR certificates determine how close the retrospective information structure comes to that reference.

\subsection*{How to shrink the radius in practice}

A large radius should not be interpreted only as failure. It identifies what information is missing and what validation would make the study more informative. If clinical time is ambiguous, the next step is source sentence and section audit. If availability is ambiguous, the next step is adjudication of whether the fact appeared in a prior note, order, structured field, or decision record before treatment. If the extractor is unstable, the next step is version locking, repeated extraction, and seed or fold variation assessment. If clinician judgement is unrecorded, the next step is chart review, prospective substudy, or an explicit sensitivity parameter. If observation processes are unresolved, visit frequency, testing intensity, note density, censoring, and follow up should be modeled as observation variables rather than absorbed into disease state. The practical output of information set emulation is therefore not only a yes or no decision. It is a map of which validation actions would reduce $\widehat B_t(K)$ and which unresolved components force reporting a compatible estimand set.

\subsection*{Downstream domination is a primary-analysis rule}

The downstream domination rule is intentionally strict for the primary point estimator. If a feature may still be a mediator, outcome proxy, post treatment descriptor, or observation process, treating it as an ordinary baseline covariate would choose one world from the fiber without evidence. The feature is not discarded. It remains available for a compatible estimand set, a mediation lock, an observation model, a mechanism analysis, or a next study hypothesis. In applied work, analysts may also report a downstream contamination sensitivity analysis. For example, let
\[
\mathfrak P_\gamma(K)=\{\Psi(M):M\in\mathcal M(K),\ P(T_j\in\mathcal A_{down}\mid j\in J_{amb})\le\gamma\},
\]
and define
\[
r_\gamma(K)=\inf_{c\in\mathbb{R}}\sup_{\psi\in\mathfrak P_\gamma(K)}|\psi-c|.
\]
The parameter $\gamma$ is a sensitivity parameter, not an estimated fact. It encodes an extra assumption about how much downstream contamination remains among ambiguous features. Such curves are useful for explaining how strong an assumption would be required to recover a narrow claim. They are not a replacement for the primary rule, and the primary point estimator should not include ambiguous features merely because one value of $\gamma$ is favorable.

Several limitations remain. First, the certificate is operational but not automatic. The rubric in Table~\ref{tab:certificate_rubric} makes the required evidence explicit, but real implementation will require trained reviewers, time aware extraction, annotation protocols, and calibration of AI systems against chart review or prospective sub studies. Human review does not remove uncertainty; it relocates disagreement into auditable validation components. When reviewers disagree about source, clock, availability, or role, the primary analysis should not resolve that disagreement by convenience. It should widen the reporting radius or route the feature to compatible reporting until additional evidence is obtained. Second, the paper uses synthetic data and synthetic text vignettes only. This is appropriate for checking theorems because the true prospective information set is known, but it does not prove feasibility in a real hospital system. The next empirical step is not a treatment-effect analysis. It is a small deidentified EHR Phase~0 study with clinical NLP that measures whether source pointers, clocks, availability claims, and causal roles can be adjudicated with acceptable residual ambiguity. Third, generated regressor uncertainty is only partially addressed. The extractor regimes in Table~\ref{tab:extractor_regimes} clarify when the current estimator is valid, but study learned extractors require fold separated extraction and additional validation. Fourth, information radius estimation depends on role enumeration and sensitivity bounds. Coverage of the proposed reporting set depends on its simultaneous sampling and uniform validation conditions; uncalibrated sensitivity allowances alone do not provide a coverage guarantee.

Future theory could study quantitative links between global experiment distance, validation error, and target specific information radius under additional structure, such as specified information projection models or smoothness restrictions on the estimand map. We do not impose such global assumptions here. The present framework deliberately keeps the primary inferential object local to the locked scalar target, where the algebraic link between certified evidence, observational fibers, compatible estimand images, and exact minimax risk remains direct and auditable.

The framework does not solve unmeasured confounding by definition. If a treatment decision depended on information absent from the EHR, the compatible estimand set should remain wide. If observation is strongly informative and cannot be modeled, the EHR compression drift will remain large. If an AI system extracts outcome adjacent summaries and labels them as baseline state, the certificate should fail. The method therefore gives both a route to causal estimation and a stopping rule.

The main contribution is the link between information reconstruction and causal effect estimation. Target trial emulation specifies the trial. Information set emulation specifies the pre treatment knowledge needed for that trial to be estimated from EHR data. When the reconstructed information set is certified, the analysis proceeds to a causal estimator. When it is not certified, the analysis reports compatible estimands and an information radius instead of pretending that a single estimand has been identified.

\section{Conclusion}

This paper introduced information set emulation for causal inference with AI and EHRs. The central claim is that emulating a target trial is not enough. A retrospective EHR study must also emulate the pre treatment information set that would have been available in the corresponding prospective study. AI is central because it can perform the initial typed lift of the raw EHR event stream: notes, orders, laboratory trajectories, encounters, and reports are converted into a typed causal information ledger before estimation begins.

The analysis turns on the set of prospective causal estimands compatible with the typed evidence produced by the AI lift. Without sufficient evidence about source time, availability, and causal role, the prospective causal effect can remain nonidentified from the flat EHR table; the fiber radius theorem shows that no estimator can overcome incompatible information set worlds that imply different estimands. With typed evidence, the compatible estimand set shrinks. The information radius measures residual target ambiguity; its square gives the exact minimax mean squared error left by that structure. The value of typing and value of certification measure how much AI typing, chart review, availability adjudication, source separation, repeated extraction, or prospective substudy reduce that distance and the corresponding exact minimax information risk.

Information set emulation is the target of the method; the information radius is its diagnostic. When the radius is small enough for the scientific tolerance, standard causal estimation can proceed using the certified reconstructed information set. When the radius remains large, the correct output is not an overconfident point estimate but a compatible reporting set and a diagnosis of what information is missing. The promise is therefore conditional but substantial: retrospective EHR studies can approach prospective study reasoning when AI first organizes EHR data into a typed information set and when the remaining information radius is small enough for the target claim.

\section*{Data and code availability}

The paper is a theory and artificial simulation study. No real patient data are used. The supplementary package contains the manuscript source, executable simulation scripts, full and script-specific seed schedules, raw replicate level results, scripts that regenerate the simulation tables and frontier figure, a verification script, synthetic Phase~0 note validation files, a certificate template, and a minimal reporting set calculator for compatible estimands. The verification script reruns the main simulation, stress tests, exact compatible-world frontier, compression-drift validation, table generation, frontier figure generation, and synthetic Phase~0 diagnostics and analysis-spread illustration, then compares regenerated outputs with the bundled results. Simulation reporting separates estimands, methods, and performance measures \citep{morris2019}. A real EHR or clinical NLP case study would require separate data access, ethics approval or exemption, and site specific annotation review.

\appendix

\section{Artificial data generating mechanisms}

This appendix records the artificial data generating mechanisms used in the simulations. The purpose is reproducibility and interpretation of the stress tests, not realism. Let
\[
X\sim N(0,1), \qquad J=0.65X+\varepsilon_J,
\]
where $J$ is latent clinician judgement and $\varepsilon_J\sim N(0,1)$. Visit or measurement intensity is
\[
O=0.35X+0.55J+\varepsilon_O,
\]
with $\varepsilon_O\sim N(0,1)$. The design or treatment availability indicator is
\[
D\sim \operatorname{Bernoulli}\{\operatorname{expit}(-0.15+0.55X-0.25O+0.25J)\}.
\]
Presence in the EHR analytic frame is
\[
S\sim \operatorname{Bernoulli}\{\operatorname{expit}(1.20+0.35X+0.55D+0.45O)\}.
\]
Treatment choice is generated as
\[
A\sim \operatorname{Bernoulli}\{\operatorname{expit}(-0.25+0.60X+0.85D+1.05J+0.60O)\}.
\]
The potential outcomes are
\[
Y(a)=0.30+0.75X+0.80D+1.00J+0.65O+a(0.60+0.30J)+\varepsilon_Y,
\]
where $\varepsilon_Y\sim N(0,1)$. The individual treatment effect is therefore
\[
Y(1)-Y(0)=0.60+0.30J.
\]
Outcome observation is generated as
\[
\Delta\sim \operatorname{Bernoulli}\{\operatorname{expit}(1.05+0.25A+0.45X+0.50J+0.75O+0.35D)\}.
\]
The structured EHR contains $X$, $D$, and $O$ with optional measurement noise, but not $J$. The primary AI reconstruction is
\[
Z_{pre}=J+e_{pre},
\]
where $e_{pre}\sim N(0,\sigma_{AI}^2)$. The source separated note feature is
\[
Z_{sep}=J+0.10X+e_{sep}.
\]
The compressed AI score that erases design and observation information is
\[
Z_{comp}=0.60X+0.85J+e_{comp}.
\]
The post index or outcome adjacent feature is generated through a mediator
\[
M=0.65A+0.55J+0.25X+0.20D+e_M
\]
and
\[
Z_{post}=0.90M+0.80Y+e_{post}.
\]
The prospective target in each replication is the empirical mean of $0.60+0.30J_i$ among patients with $S_i=1$.

The additional stress tests modify this baseline mechanism. Extraction error varies $\sigma_{AI}$. Partial leakage replaces $Z_{pre}$ by a standardized blend $Z_{pre}+\lambda Z_{post}$. Certificate role misclassification replaces a proportion of certified baseline features by post index summaries. Availability error replaces a pre clinical judgement signal by a late source that combines pre treatment signal with treatment and outcome adjacent information. Positivity stress multiplies the treatment logit coefficients by a factor larger than one.

\section{Certificate and reporting templates}

A machine readable certificate template should contain at least the following fields:
\[
\begin{gathered}
\{\text{feature\_id},\ \text{source\_pointer},\ \text{feature\_value},\ \text{clinical\_time},\ \\
\text{measurement\_time},\ \text{recording\_time},\ \text{availability\_time},\ \\
\text{treatment\_relation},\ \text{outcome\_relation},\ \text{proposed\_roles},\ \\
\text{compatible\_roles},\ \text{certified\_roles},\ \text{failed\_predicates},\ \\
\text{representation\_version},\ 
\text{seed\_rule},\ \text{reviewer\_status}\}.
\end{gathered}
\]
A minimal reporting set calculator takes a table of role specific estimates $\widehat\Psi_t$, standard errors $\widehat{se}_t$, and validation-calibrated sensitivity envelopes $\widehat B_t(K)$ and returns
\[
\widehat{\mathfrak P}_{1-\alpha}(K)=\bigcup_{t\in\mathcal C(K)}
[\widehat\Psi_t-c_S\widehat{se}_t-\widehat B_t(K),
\widehat\Psi_t+c_S\widehat{se}_t+\widehat B_t(K)].
\]
The calculator uses $c_S=z_{1-\alpha_S/(2m)}$, where $m=|\mathcal C(K)|$ and $\alpha_S+\alpha_V\le\alpha$, and preserves disconnected interval unions. Every assignment must concern the same locked scalar target. Image coverage requires the simultaneous sampling and uniform validation conditions in the calibration proposition; the calculator does not establish those conditions. In a real study, every nonzero bound should be linked to a chart review, repeated extraction audit, sensitivity parameter, overlap diagnostic, or observation process analysis.

\section{A conservative MNAR observation bound}

This appendix gives a simple operational bound for the observation sensitivity family in the main text. It is not a replacement for a problem-specific missingness analysis. Suppose the outcome is reported on a bounded scale $[y_{min},y_{max}]$, or is explicitly clipped to a prespecified range that defines the reported target, and let $R_Y=y_{max}-y_{min}$. If, within treatment arm $a$, the unmeasured observation mechanism can tilt the baseline observed-outcome distribution by a density ratio bounded between $\exp(-\Gamma_a)$ and $\exp(\Gamma_a)$, then the absolute change in an arm-specific mean is bounded conservatively by
\[
B_{MNAR,a}(\Gamma_a)=R_Y\tanh(\Gamma_a/2).
\]
For a treatment contrast, one may add
\[
B_{MNAR}(\Gamma)=B_{MNAR,1}(\Gamma_1)+B_{MNAR,0}(\Gamma_0)
\]
to the observation component of the reporting bound. For unbounded outcomes, a standardized scale or a plausible range alone does not establish this bound. One must either assume that the declared support bound is valid or explicitly change the target to a prespecified clipped outcome; otherwise a separate tail or moment bound is needed. The purpose is not to make MNAR harmless; it is to keep the unverified observation assumption visible on the same estimand scale as the information radius.


\begin{thebibliography}{99}

\bibitem[Angelopoulos et al.(2023)]{angelopoulos2023}
Angelopoulos, A. N., Bates, S., Fannjiang, C., Jordan, M. I., and Zrnic, T. Prediction powered inference. \emph{Science}. 2023;382:669 to 674.

\bibitem[Battaglia et al.(2024)]{battaglia2024}
Battaglia, L., Christensen, T., Hansen, S., and Sacher, S. Inference for regression with variables generated by AI or machine learning. arXiv:2402.15585. 2024.

\bibitem[Cashin et al.(2025)]{target2025}
Cashin, A. G., et al. Transparent reporting of observational studies emulating a target trial: the TARGET Statement. \emph{JAMA}. 2025;334(12):1084--1093. doi:10.1001/jama.2025.13350.

\bibitem[Chernozhukov et al.(2018)]{dml2018}
Chernozhukov, V., Chetverikov, D., Demirer, M., Duflo, E., Hansen, C., Newey, W., and Robins, J. Double debiased machine learning for treatment and structural parameters. \emph{The Econometrics Journal}. 2018;21:C1 to C68.


\bibitem[Clivio et al.(2026)]{clivio2026}
Clivio, O., D'Amour, A., Franks, A., Bruns-Smith, D., Holmes, C., and Feller, A. Deconfounding scores and representation learning for causal effect estimation with weak overlap. arXiv:2604.00811v1. 2026.

\bibitem[D'Amour and Franks(2021)]{damour2021}
D'Amour, A. and Franks, A. Deconfounding scores: feature representations for causal effect estimation with weak overlap. arXiv:2104.05762v1. 2021.

\bibitem[Daoud et al.(2022)]{daoud2022}
Daoud, A., Jerzak, C. T., and Johansson, R. Conceptualizing treatment leakage in text-based causal inference. In \emph{Proceedings of the 2022 Conference of the North American Chapter of the Association for Computational Linguistics: Human Language Technologies}. 2022:5638--5645. doi:10.18653/v1/2022.naacl-main.413.

\bibitem[Fong and Tyler(2021)]{fongtyler2021}
Fong, C. and Tyler, M. Machine learning predictions as regression covariates. \emph{Political Analysis}. 2021;29:467 to 484.

\bibitem[Franklin et al.(2021)]{rctduplicate2021}
Franklin, J. M., et al. Emulating randomized clinical trials with nonrandomized real-world evidence studies: first results from the RCT-DUPLICATE initiative. \emph{Circulation}. 2021;143(10):1002--1013. doi:10.1161/CIRCULATIONAHA.120.051718.

\bibitem[Hernan and Robins(2020)]{hernanrobins2020}
Hernan, M. A. and Robins, J. M. \emph{Causal Inference: What If}. Boca Raton: Chapman and Hall CRC. 2020.

\bibitem[Lee et al.(2016)]{lee2016}
Lee, J. D., Sun, D. L., Sun, Y., and Taylor, J. E. Exact post selection inference, with application to the lasso. \emph{Annals of Statistics}. 2016;44:907 to 927.

\bibitem[Liu et al.(2026)]{liu2026}
Liu, L., Chen, J., and Macropol, K. LLM-extracted covariates for clinical causal inference: rethinking integration strategies. arXiv:2604.16763v3, 2026.

\bibitem[Mozer et al.(2024)]{mozer2024}
Mozer, R., Kaufman, A. R., Celi, L. A., and Miratrix, L. Leveraging text data for causal inference using electronic health records. arXiv:2307.03687, revised 2024.


\bibitem[Morris et al.(2019)]{morris2019}
Morris, T. P., White, I. R., and Crowther, M. J. Using simulation studies to evaluate statistical methods. \emph{Statistics in Medicine}. 2019;38:2074 to 2102.

\bibitem[Rafalko et al.(2025)]{rafalko2025}
Rafalko, N., Gianfrancesco, M., and Goldstein, N. D. On the use of natural language processing to implement the target trial framework using unstructured data from the electronic health record. \emph{Global Epidemiology}. 2025;9:100204. doi:10.1016/j.gloepi.2025.100204.

\bibitem[Robins et al.(2000)]{robins2000}
Robins, J. M., Hernan, M. A., and Brumback, B. Marginal structural models and causal inference in epidemiology. \emph{Epidemiology}. 2000;11:550 to 560.

\bibitem[Veitch et al.(2020)]{veitch2020}
Veitch, V., Sridhar, D., and Blei, D. Adapting text embeddings for causal inference. \emph{Proceedings of UAI}. 2020.

\bibitem[Wang et al.(2023)]{wang2023}
Wang, S. V., Schneeweiss, S., and the RCT-DUPLICATE Initiative. Emulation of randomized clinical trials with nonrandomized database analyses: results of 32 clinical trials. \emph{JAMA}. 2023;329(16):1376--1385. doi:10.1001/jama.2023.4221.

\bibitem[Wang et al.(2026)]{wang2026}
Wang, Y., Li, Y., Lin, T., et al. An operational target trial emulation framework for causal inference using electronic health record data. \emph{npj Digital Medicine}. 2026;9:424. doi:10.1038/s41746-026-02563-z.

\bibitem[Zeng et al.(2022)]{zeng2022}
Zeng, J., et al. Uncovering interpretable potential confounders in electronic medical records. \emph{Nature Communications}. 2022;13:1014. doi:10.1038/s41467-022-28546-8.


\bibitem[Le Cam(1986)]{lecam1986}
Le Cam, L. \emph{Asymptotic Methods in Statistical Decision Theory}. Springer; 1986.

\end{thebibliography}
\end{document}